\documentclass{bmcart}

\usepackage{microtype,color,graphicx}
\usepackage[utf8]{inputenc} 
\usepackage{hyperref}
\usepackage{amsmath}
\usepackage{amsthm}
\usepackage{subfig}
\newtheorem{lemma}{Lemma}
\newtheorem{theorem}{Theorem}
\usepackage{float}
\floatstyle{plain}
\restylefloat{figure}

\newcommand{\BWT}{\ensuremath{\mathrm{BWT}}}

\newcommand{\SA}{\ensuremath{\mathrm{SA}}}
\newcommand{\LF}{\ensuremath{\mathrm{LF}}}

\begin{document}
\begin{frontmatter}
\begin{fmbox}
\dochead{Research}

\title{Prefix-Free Parsing for Building Big BWTs}

\author[
   addressref={aff1},  email={christinadotboucher@gmail.com}
]{\inits{C}\fnm{Christina} \snm{Boucher}}
\author[
   addressref={aff2,aff3},   email={travis.gagie@gmail.com}
]{\inits{T}\fnm{Travis} \snm{Gagie}}
\author[
   addressref={aff1,aff4},   email={kuhnle@ufl.edu}
]{\inits{A}\fnm{Alan} \snm{Kuhnle}}
\author[
addressref={aff5}, email={ben.langmead@gmail.com}
]{\inits{B}\fnm{Ben} \snm{Langmead}}
\author[
   addressref={aff6,aff7}, email={giovanni.manzini@uniupo.it}
]{\inits{G}\fnm{Giovanni} \snm{Manzini}}
\author[
   addressref={aff5}, email={tmun1@jhu.edu}
]{\inits{T}\fnm{Taher} \snm{Mun}}

\address[id=aff1]{
  \orgname{CISE, University of Florida}, 
  \city{Gainesville},                              
  \state{FL}
  \cny{USA}                                    
}

\address[id=aff2]{
  \orgname{EIT, Diego Portales University}, 
  \city{Santiago},                              
  \cny{Chile}                                    
}

\address[id=aff3]{
  \orgname{CeBiB}, 
  \city{Santiago},                              
  \cny{Chile}                                    
}

\address[id=aff4]{
  \orgname{Informatics Institute}, 
  \city{Gainesville},                              
  \state{FL}
  \cny{USA}                                    
}

\address[id=aff5]{
  \orgname{Johns Hopkins University}, 
  \city{Baltimore},                              
  \state{MD},
  \cny{USA}                                    
}

\address[id=aff6]{
  \orgname{University of Eastern Piedmont}, 
  \city{Alessandria},                              
  \cny{Italy}                                    
}

\address[id=aff7]{
  \orgname{IIT, CNR}, 
  \city{Pisa},                              
  \cny{Italy}                                    
}
\end{fmbox}

\begin{abstractbox}
\begin{abstract}
    High-throughput sequencing technologies have led to explosive growth of genomic databases; one of which will soon reach hundreds of terabytes.  For many applications we want to build and store indexes of these databases but constructing such indexes is a challenge.  Fortunately, many of these genomic databases are highly-repetitive---a characteristic that can be exploited to ease the computation of the Burrows-Wheeler Transform (BWT), which underlies many popular indexes.  In this paper, we introduce a preprocessing algorithm, referred to as {\em prefix-free parsing}, that takes a text $T$ as input, and in one-pass generates a dictionary $D$ and a parse $P$ of $T$ with the property that the BWT of $T$ can be constructed from $D$ and $P$ using workspace proportional to their total size and $O (|T|)$-time.  Our experiments show that $D$ and $P$ are significantly smaller than $T$ in practice, and thus, can fit in a reasonable internal memory even when $T$ is very large.  In particular, we show that with prefix-free parsing we can build an 131-megabyte run-length compressed FM-index (restricted to support only counting and not locating) for 1000 copies of human chromosome 19 in 2 hours using 21 gigabytes of memory, suggesting that we can build a 6.73 gigabyte index for 1000 complete human-genome haplotypes in approximately 102 hours using about 1 terabyte of memory.


\end{abstract}
\end{abstractbox}

\end{frontmatter}

\section{Introduction}
\label{sec:introduction}

The money and time needed to sequence a genome have shrunk shockingly quickly and researchers' ambitions have grown almost as quickly: the Human Genome Project cost billions of dollars and took a decade but now we can sequence a genome for about a thousand dollars in about a day.  The 1000 Genomes Project \cite{1000genomes} was announced in 2008 and completed in 2015, and now the 100,000 Genomes Project is well under way \cite{100K}.  With no compression 100,000 human genomes occupy roughly 300 terabytes of space, and genomic databases will have grown even more by the time a standard research machine has that much RAM.   At the same time, other initiatives have began to study how microbial species behave and thrive in environments.  These initiatives are generating public datasets, which are larger than the 100,000 Genomes Project.   For example, in recent years, there has been an initiative to move toward using whole genome sequencing to accurately identify and track foodborne pathogens (e.g. antibiotic-resistant bacteria)~\cite{carleton2016whole}. This led to the  GenomeTrakr initiative, which is a large public effort to use genome sequencing for surveillance and detection of outbreaks of foodborne illnesses. Currently, GenomeTrakr includes over 100,000 samples, spanning several species available through this initiative---a number that continues to rise as datasets are continually added \cite{genometrakr}.  Unfortunately, analysis of this data is limited due to their size, even though the similarity between genomes of individuals of the same species means the data is highly compressible.

These public databases are used in various applications --- e.g., to detect genetic variation within individuals, determine evolutionary history within a population, and assemble the genomes of novel (microbial) species or genes.  Pattern matching within these large databases is fundamental to all these applications, yet repeatedly scanning these --- even compressed --- databases is infeasible.  Thus, for these and many other applications, we want to build and use indexes from the database.  Since these indexes should fit in RAM and cannot rely on word boundaries, there are only a few candidates.  Many of the popular indexes in bioinformatics are based on the Burrows-Wheeler Transform (BWT)~\cite{BW94} and there have been a number of papers about building BWTs for genomic databases, e.g.,~\cite{Sir16} and references therein.  However, it is difficult to process anything more than a few terabytes of raw data per day with current techniques and technology because of the difficulty of working in external memory.

Since genomic databases are often highly repetitive, we revisit the idea of applying a simple compression scheme and then computing the BWT from the resulting encoding in internal memory.  This is far from being a novel idea --- e.g., Ferragina, Gagie and Manzini's {\tt bwtdisk} software~\cite{FGM12} could already in 2010 take advantage of its input being given compressed, and Policriti and Prezza~\cite{PP17} showed how to compute the BWT from the LZ77 parse of the input using $O (n (\log r + \log z))$-time and $O (r + z)$-space, where $n$ is the length of the uncompressed input, $r$ is the number of runs in the BWT and $z$ is the number of phrases in the LZ77 parse --- but we think the preprocessing step we describe here, {\em prefix-free parsing}, stands out because of its simplicity and flexibility.   Once we have the results of the parsing, which are a dictionary and a parse, building the BWT out of them is more involved, yet when our approach works well, the dictionary and the parse are together much smaller than the initial dataset and that makes the BWT computation less resource-intensive.

{\bf Our Contributions.} In this paper, we formally define and present prefix-free parsing. The main idea of this method is to divide the input text into overlapping variable-length phrases with delimiting prefixes and suffixes.   To accomplish this division, we slide a window of length $w$ over the text and, whenever the Karp-Rabin hash of the window is 0 modulo $p$, we terminate the current phrase at the end of the window and start the next one at the beginning of the window.  This concept is partly inspired by {\tt rsync}'s~\cite{rsync} use of a rolling hash for content-slicing.  Here, $w$ and $p$ are parameters that affect the size of the dictionary of distinct phrases and the number of phrases in the parse.  This takes linear-time and one pass over the text, or it can be sped up by running several windows in different positions over the text in parallel and then merging the results.

Just as {\tt rsync} can usually recognize when most of a file remains the same, we expect that for most genomic databases and good choices of $w$ and $p$, the total length of the phrases in the dictionary and the number of phrases in the parse will be small in comparison to the uncompressed size of the database.  We demonstrate experimentally that with prefix-free parsing we can compute BWT using less memory and equivalent time.  In particular, using our method we reduce peak memory usage up to 10x over a standard baseline algorithm which computes the BWT by first computing the suffix array using the algorithm SACA-K~\cite{tois/Nong13}, while requiring roughly the same time on large sets of salmonella genomes obtained from GenomeTrakr. 

In Section~\ref{sec:theory}, we show how we can compute the BWT of the text from the dictionary and the parse alone using workspace proportional only to their total size, and time linear in the uncompressed size of the text when we can work in internal memory.  In Section~\ref{sec:practice} we describe our implementation and report the results of our experiments showing that in practice the dictionary and parse often are significantly smaller than the text and so may fit in a reasonable internal memory even when the text is very large, and that this often makes the overall BWT computation both faster and smaller.  In Section~\ref{sec:indexing} we describe how we build run-length compressed FM-indexes~\cite{FM05} (which only support counting and not locating) for datasets consisting of 50, 100, 200 and 500 using prefix-free parsing.  Our results suggest that we can build a roughly 6.73-gigabyte index for 1000 complete human genomes in about 102 hours using about 1.1 terabytes of memory.
Prefix-free parsing and all accompanied documents are available at \url{https://gitlab.com/manzai/Big-BWT}.

\section{Review of the Burrows-Wheeler Transform}
\label{app:bwt}

As part of the Human Genome Project, researchers had to piece together a huge number of relatively tiny, overlapping pieces of DNA, called reads, to assemble a reference genome about which they had little prior knowledge.  Once the Project was completed, however, they could then use that reference genome as a guide to assemble other human genomes much more easily.  To do this, they indexed the reference genome such that, after running a DNA sample from a new person through a sequencing machine and obtaining another collection of reads, for each of those new reads they could quickly determine which part of the reference genome it matched most closely.  Since any two humans are genetically very similar, aligning the new reads against the reference genome gives a good idea of how they are really laid out in the person's genome.

In practice, the best solutions to this problem of indexed approximate matching work by reducing it to a problem of indexed exact matching, which we can formalize as follows: given a string $T$ (which can be the concatenation of a collection of strings, terminated by special symbols), pre-process it such that later, given a pattern $P$, we can quickly list all the locations where $P$ occurs in $T$.  We now start with a simple but impractical solution to the latter problem, and then refine it until we arrive at a fair approximation of the basis of most modern assemblers, illustrating the workings of the Burrows-Wheeler Transform (BWT) along the way.

Suppose we want to index the three strings {\rm GATTACAT}, {\rm GATACAT} and {\rm GATTAGATA}, so $T [0..n -1] = \mathrm{GATTACAT\$_1GATACAT\$_2GATTAGATA\$_3}$, where $\$_1$, $\$_2$ and $\$_3$ are terminator symbols.  Perhaps the simplest solution to the problem of indexing $T$ is to build a trie of the suffixes of the three strings in our collection (i.e., an edge-labelled tree whose root-to-leaf paths are the suffixes of those strings) with each leaf storing the starting position of the suffix labelling the path to that leaf, as shown in Figure~\ref{fig:trie}.

\begin{figure}
\begin{center}
\includegraphics[width=\textwidth]{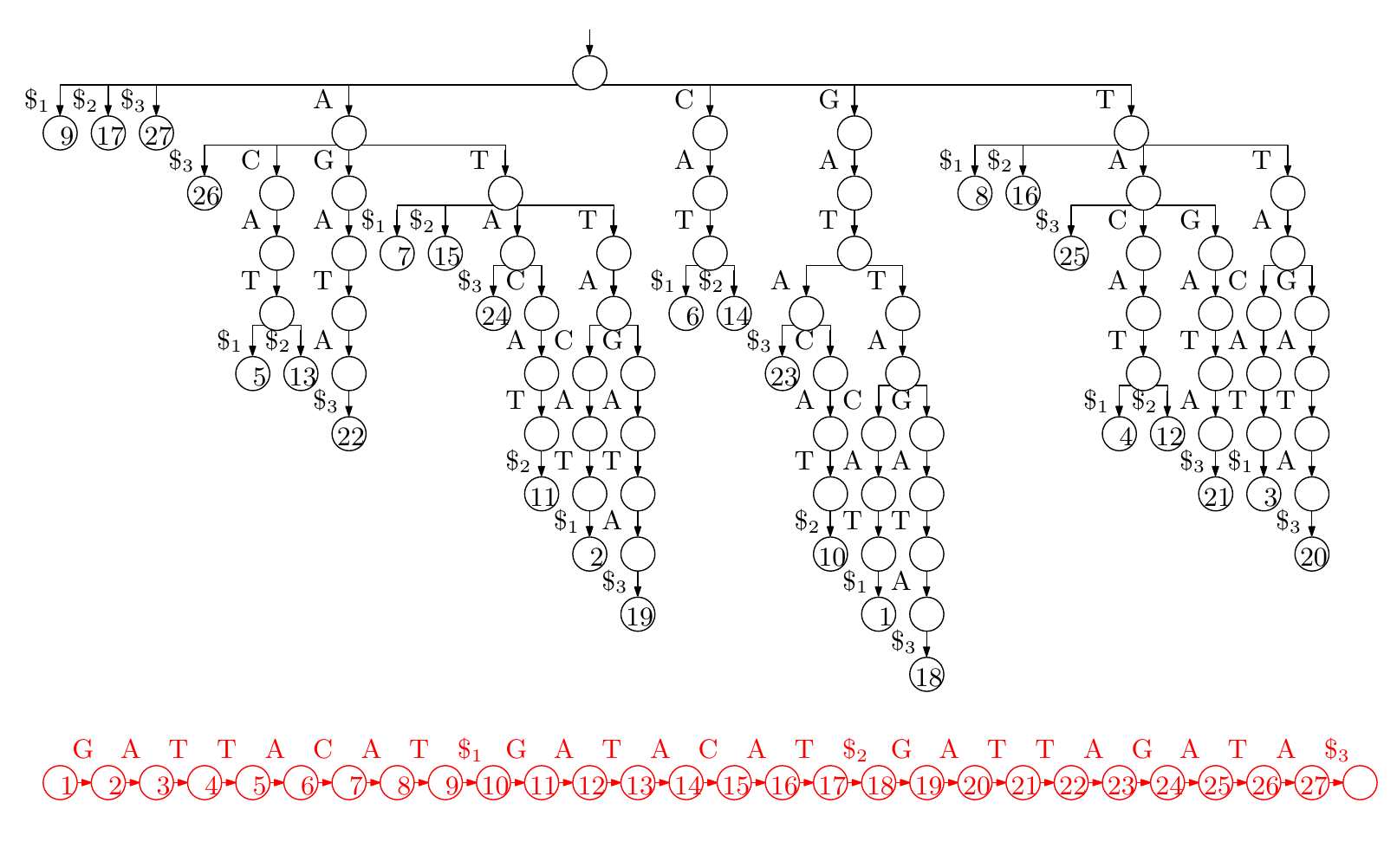}
\caption{The suffix trie for our example with the three strings {\rm GATTACAT}, {\rm GATACAT} and {\rm GATTAGATA}.  The input is shown at the bottom, in red because we do not need to store it.}
\label{fig:trie}
\end{center}
\end{figure}

Suppose every node stores pointers to its children and its leftmost and rightmost leaf descendants, and every leaf stores a pointer to the next leaf to its right.  Then given $P [0..m-1]$, we can start at the root and descend along a path (if there is one) such that the label on the edge leading to the node at depth $i$ is $P [i - 1]$, until we reach a node $v$ at depth $m$.  We then traverse the leaves in $v$'s subtree, reporting the the starting positions stored at them, by following the pointer from $v$ to its leftmost leaf descendant and then following the pointer from each leaf to the next leaf to its right until we reach $v$'s rightmost leaf descendant.

The trie of the suffixes can have a quadratic number of nodes, so it is impractical for large strings.  If we remove nodes with exactly one child (concatenating the edge-labels above and below them), however, then there are only linearly many nodes, and each edge-label is a substring of the input and can be represented in constant space if we have the input stored as well.  The resulting structure is essentially a suffix tree (although it lacks suffix and Weiner links), as shown in Figure~\ref{fig:tree}.  Notice that the label of the path leading to a node $v$ is the longest common prefix of the suffixes starting at the positions stored at $v$'s leftmost and rightmost leaf descendants, so we can navigate in the suffix tree, using only the pointers we already have and access to the input.

\begin{figure}
\begin{center}
\includegraphics[width=\textwidth]{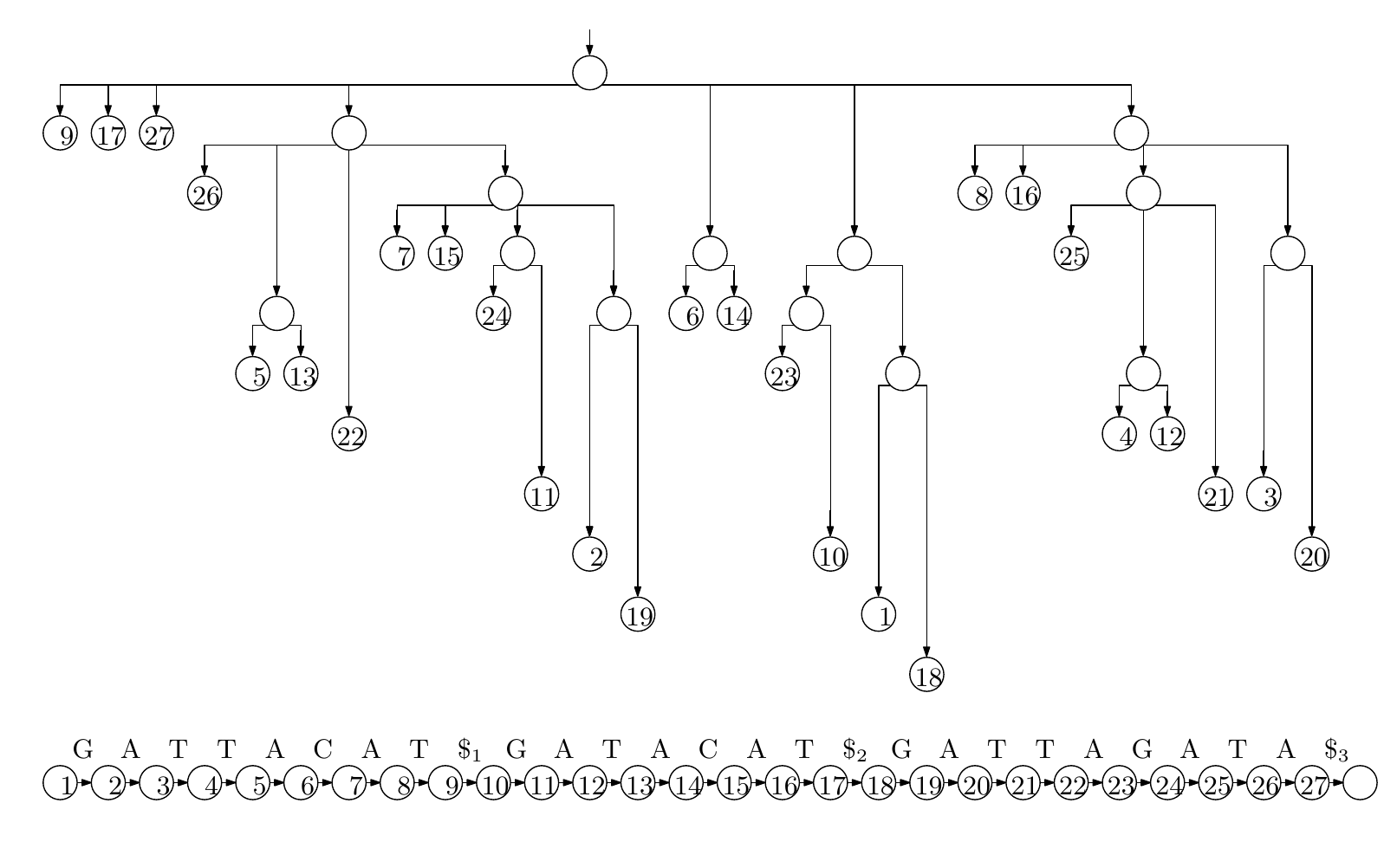}
\caption{The suffix tree for our example.  We now also need to store the input.}
\label{fig:tree}
\end{center}
\end{figure}

Although linear, the suffix tree still takes up an impractical amount of space, using several bytes for each character of the input.  This is significantly reduced if we discard the shape of the tree, keeping only the input and the starting positions in an array, which is called the suffix array (SA).  The SA for our example is shown in Figure~\ref{fig:array}.  Since the entries of the SA are the starting points of the suffixes in lexicographic order, with access to $T$ we can perform two binary searches to find the endpoints of the interval of the suffix array containing the starting points of suffixes starting with $P$: at each step, we consider an entry $\SA [i]$ and check if $T [\SA [i]]$ lexicographically precedes $P$.  This takes a total of $O (m \log n)$ time done na\"ively, and can be sped up with more sophisticated searching and relatively small auxiliary data structures.

\begin{figure}
\begin{center}
\includegraphics[width=\textwidth]{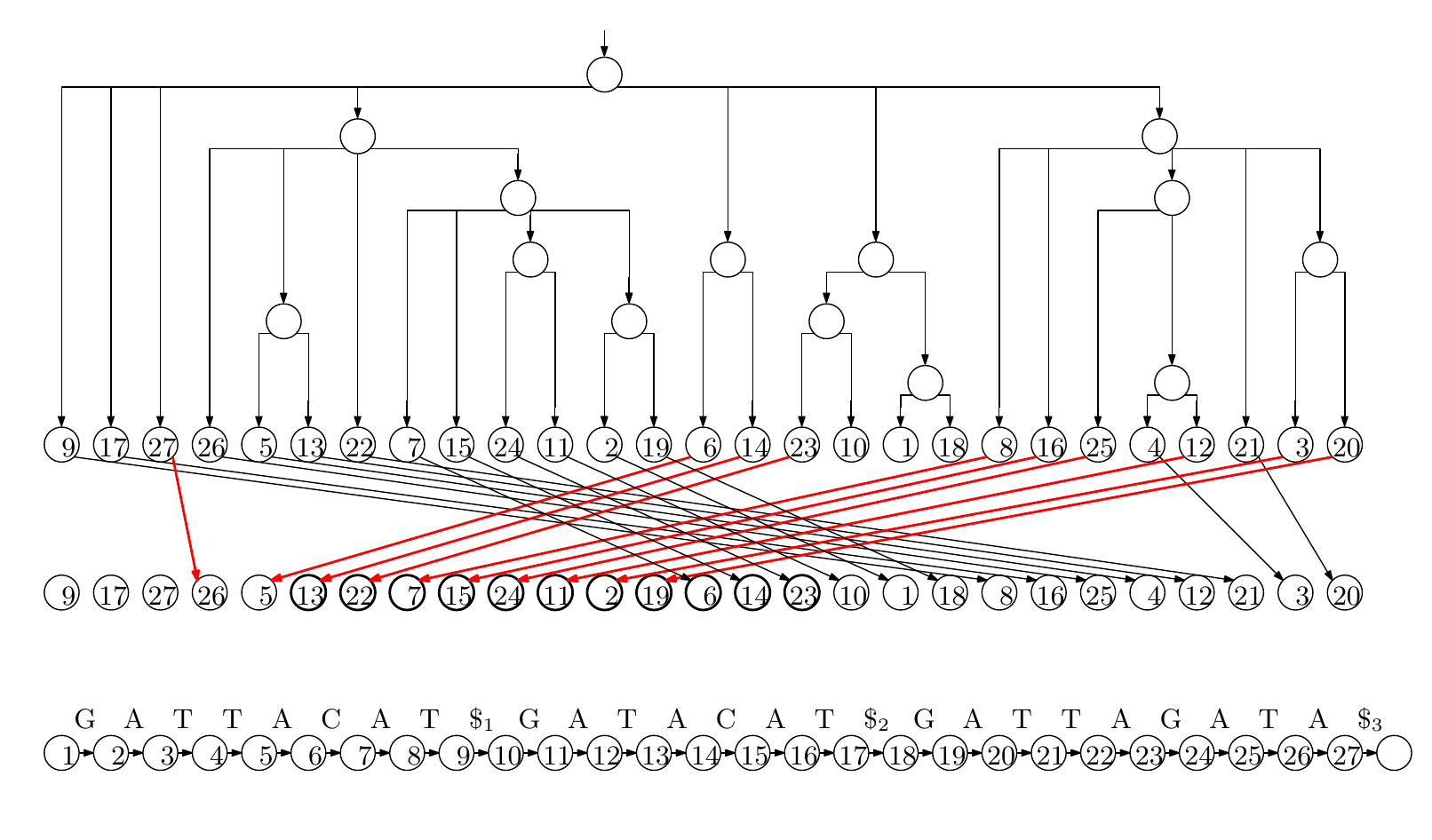}
\caption{The suffix array for our example is the sequence of values stored in the leaves of the tree (which we need not store explicitly).  The LF mapping is shown as the arrows between two copies of the suffix array; the arrows to values $i$ such that $T [\SA [i]] = \mathtt{A}$ are in red, to illustrate that they point to consecutive positions in the suffix array and do not cross.  Since $\Psi$ is the inverse of the LF mapping, it can be obtained by simply reversing the direction of the arrows.}
\label{fig:array}
\end{center}
\end{figure}

Even the SA takes linear space, however, which is significantly more than what is needed to store the input when the alphabet is small (as it is in the case of DNA).  Let $\Psi$ be the function that, given the position of a value $i < n - 1$ in the SA, returns the position of $i + 1$.  Notice that, if we write down the first character of each suffix in the order they appear in the SA, the result is a sorted list of the characters in $T$, which can be stored using using $O (\log n)$ bits for each character in the alphabet.  Once we have this list stored, given a position $i$ in SA, we can return $T [\SA [i]]$ efficiently.

Given a position $i$ in SA and a way to evaluate $\Psi$, we can extract $T [SA [i]..n - 1]$ by writing $T [\SA [i]], T [\SA [\Psi (i)]], T [\SA [\Psi^2 (i)]], \ldots$.  Therefore, we can perform the same kind of binary search we use when with access to a full suffix array.  Notice that if $T [\SA [i]] \prec T [\SA [i + 1]]$ then $\Psi (i) < \Psi (i + 1)$, meaning that $\Psi (1), \ldots, \Psi (n)$ can be divided into $\sigma$ increasing consecutive subsequences, where $\sigma$ is the size of the alphabet.   Here, $\prec$ denotes lexicographic precedence. It follows that we can store $n H_0 (T) + o (n \log \sigma)$ bits, where $H_0 (T)$ is the 0th-order empirical entropy of $T$, such that we can quickly evaluate $\Psi$.  This bound can be improved with a more careful analysis.

Now suppose that instead of a way to evaluate $\Psi$, we have a way to evaluate quickly its inverse, which is called the last-to-first (LF) mapping.  (This name was not chosen because, if we start with the position of $n$ in the suffix array and repeatedly apply the LF mapping we enumerate the positions in the SA in decreasing order of their contents, ending with 1; to some extent, the name is a lucky coincidence.)  The LF mapping for our example is also shown with arrows in Figure~\ref{fig:array}.  Since it is the inverse of $\Psi$, the sequence $\LF (1), \ldots, \LF (n)$ can be partitioned into $\sigma$ incrementing subsequences: for each character $c$ in the alphabet, if the starting positions of suffixes preceded by copies of $c$ are stored in $\SA [j_1], \ldots, \SA [j_t]$ (appearing in that order in the SA), then $\LF (j_1)$ is 1 greater than the number of characters lexicographically less than $c$ in $T$ and $\LF (j_2), \ldots, \LF (j_t)$ are the next $t - 1$ numbers.  Figure~\ref{fig:array} illustrates this, 
with heavier arrows to values $i$ such that $T [\SA [i]] = \mathrm{A}$, to illustrate that they point to consecutive positions in the suffix array and do not cross.

Consider the interval $I_{P [i..m-1]}$ of the SA containing the starting positions of suffixes beginning with $P [i..m-1]$, and the interval $I_{P [i - 1]}$ containing the starting positions of suffixes beginning with $P [i - 1]$.  If we apply the LF mapping to the SA positions in $I_{P [i..m -1]-1}$, the SA positions we obtain that lie in $I_{P [i - 1]}$ for a consecutive subinterval, containing the starting positions in $T$ of suffixes beginning with $P [i - 1..m-1]$.  Therefore, we can search also with the LF mapping.

If we write the character preceding each suffix of $T$ (considering it to be cyclic) in the lexicographic order of the suffixes, the result is the Burrows-Wheeler Transform (BWT) of $T$.  A rank data structure over the BWT (which, given a character and a position, returns the number of occurrences of that character up to that position) can be used to implement searching with the LF-mapping, together with an array $C$ indicating for each character in the alphabet how many characters in $T$ are lexicographically strictly smaller than it.  Specifically,
\[\LF (i) = \BWT.\mathrm{rank}_{\BWT [i]} (i) + C [\BWT [i]]\,.\]

If follows that, to compute $I_{P [i - 1..m-1]}$ from $I_{P [i..m-1]}$, we perform a rank query for $P [i - 1]$ immediately before the beginning of $I_{P [i..m-1]}$ and add $C [P [i + 1]] + 1$ to the result, to find the beginning of $I_{P [i - 1..m-1]}$; and we perform a rank query for $P [i - 1]$ at the end of $I_{P [i..m-1]}$ and add $C [P [i + 1]]$ to the result, to find the end of $I_{P [i - 1..m-1]}$.  Figure~\ref{fig:BWT} shows the BWT for our example, and the sorted list of characters in $T$.  Comparing it to Figure~\ref{fig:array} makes the formula above clear: if $\BWT [i]$ is the $j$th occurrence of that character in the BWT, then the arrow from $\LF (i)$ leads from $i$ to the position of the $j$th occurrence of that character in the sorted list.  This is the main idea behind FM-indexes~\cite{FM05}, and the main motivation for bioinformaticians to be interested in building BWTs.

\begin{figure}
\begin{center}
\includegraphics[width=\textwidth]{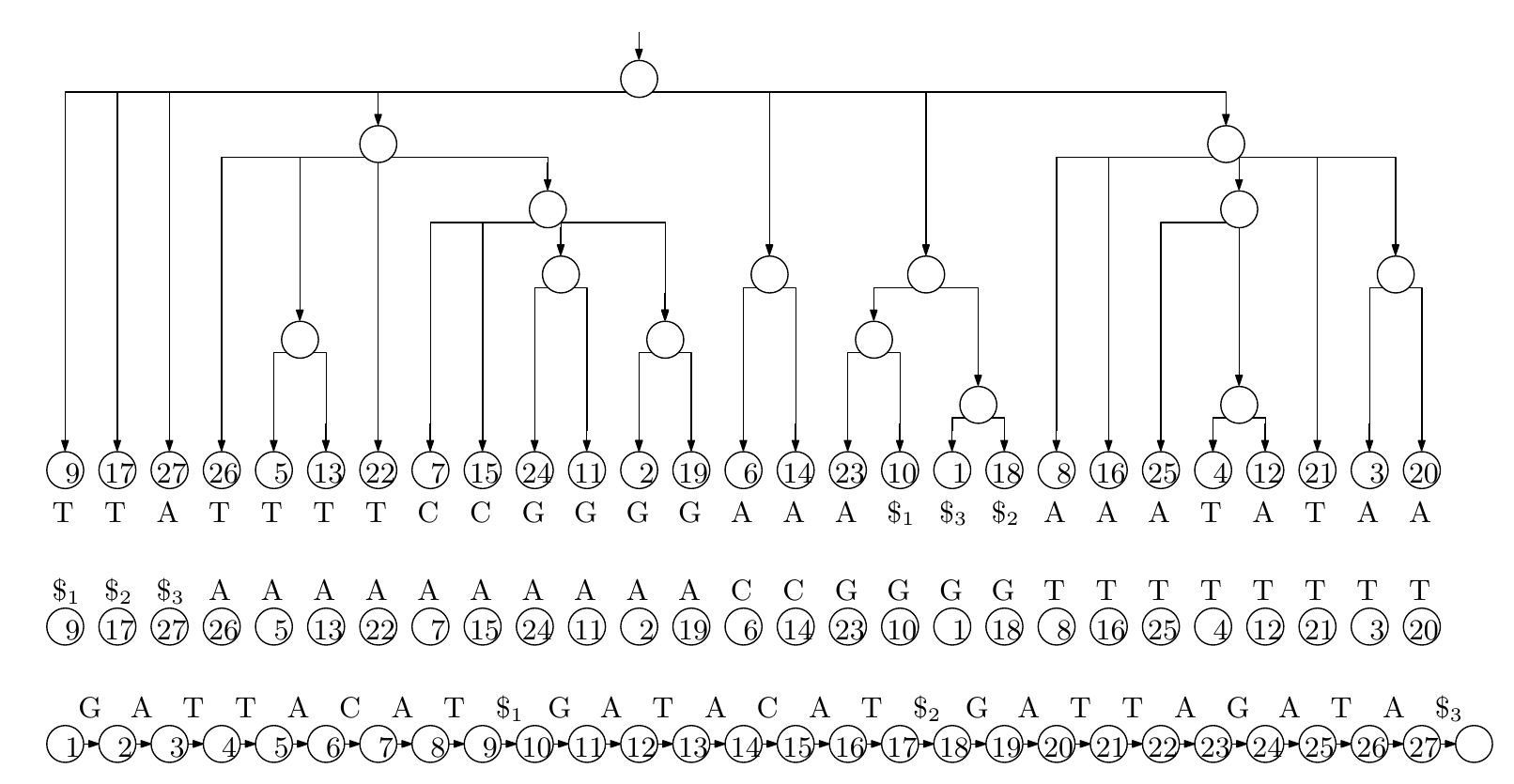}
\caption{The BWT and the sorted list of characters for our example.  Drawing arrows between corresponding occurrences of characters in the two strings gives us the diagram for the LF-mapping.}
\label{fig:BWT}
\end{center}
\end{figure}

\section{Theory of Prefix Free Parsing}
\label{sec:theory}

We let $E \subseteq \Sigma^w$ be any set of strings each of length $w \geq 1$ over the alphabet $\Sigma$ and let $E' = E \cup \{\mathtt{\#}, \mathtt{\$}^w\}$, where $\mathtt{\#}$ and $\mathtt{\$}$ are special symbols lexicographically less than any in $\Sigma$.  We consider a text $T [0..n - 1]$ over $\Sigma$ and let $D$ be the maximum set such that for $d \in D$ the following conditions hold
\begin{itemize}
\item $d$ is a substring of $\mathtt{\#}\,T\,\mathtt{\$}^w$,
\item exactly one proper prefix of $d$ is in $E'$,
\item exactly one proper suffix of $d$ is in $E'$,
\item no other substring of $d$ is in $E'$.
\end{itemize}

Given $T$ and a way to recognize strings in $E$, we can build $D$ iteratively by scanning $\mathtt{\#}\,T\,\mathtt{\$}^w$ to find occurrences of elements of $E'$, and adding to $D$ each substring of $\mathtt{\#}\,T\,\mathtt{\$}^w$ that starts at the beginning of one such occurrence and ends at the end of the next one.  While we are building $D$ we also build a list $P$ of the occurrences of the elements of $D$ in $T$, which we call the parse (although each consecutive pair of elements overlap by $w$ characters, so $P$ is not a partition of the characters of $\mathtt{\#}\,T\,\mathtt{\$}^w$).  In $P$ we identify each element of $D$ with its lexicographic rank. 

For example, suppose we have $\Sigma = \{\mathtt{!}, \mathtt{A}, \mathtt{C}, \mathtt{G}, \mathtt{T}\}$, $w = 2$, $E = \{\mathtt{AC}, \mathtt{AG}, \mathtt{T!}\}$ and
\[T = \mbox{\tt GATTACAT!GATACAT!GATTAGATA}\,.\]
Then, we get
$$
D  = \{\mathtt{\#GATTAC}, \mathtt{ACAT!}, \mathtt{AGATA\$\$}, \mathtt{T!GATAC}, \mathtt{T!GATTAG}\},
$$
the parse of $\mathtt{\#}\,T\,\mathtt{\$}^w$ is 
$$
\mathtt{\#GATTAC}\; \mathtt{ACAT!}\; \mathtt{T!GATAC}\; \mathtt{ACAT!}\; \mathtt{T!GATTAG}\; 
\mathtt{AGATA\$\$}
$$
and, identifying elements of $D$ by their lexicographic ranks, the resulting array $P$ is $P = [0, 1, 3, 1, 4, 2]$.

Next, we define $S$ as the set of suffixes of length greater than $w$ of elements of $D$. In our previous example we get
\begin{eqnarray*}
S & = & \{\mathtt{\#GATTAC}, \mathtt{GATTAC}, \ldots, \mathtt{TAC},\\
     && \mathtt{ACAT!}, \mathtt{CAT!}, \mathtt{AT!},\\
     && \mathtt{AGATA\$\$}, \mathtt{GATA\$\$}, \ldots, \mathtt{A\$\$},\\
     && \mathtt{T!GATAC}, \mathtt{!GATAC}, \ldots, \mathtt{TAC},\\
     && \mathtt{T!GATTAG}, \mathtt{!GATTAG}, \ldots, \mathtt{TAG}\}.
\end{eqnarray*}

\begin{lemma}\label{lem:prefix-free}
$S$ is a prefix-free set.
\end{lemma}

\begin{proof}
If $s \in S$ were a proper prefix of $s' \in S$ then, since $|s| > w$, the last $w$ characters of $s$ --- which are an element of $E'$ --- would be a substring of $s'$ but neither a proper prefix nor a proper suffix of $s'$.  Therefore, any element of $D$ with $s'$ as a suffix would contain at least three substrings in $E'$, contrary to the definition of $D$.
\end{proof}

\begin{lemma}
\label{lem:comparable}
Suppose $s, s' \in S$ and $s \prec s'$.  Then $s x \prec s' x'$ for any strings $x, x' \in (\Sigma \cup \{\#, \$\})^*$.
\end{lemma}

\begin{proof}
By Lemma~\ref{lem:prefix-free}, $s$ and $s'$ are not proper prefixes of each other.  Since they are not equal either (because $s \prec s'$), it follows that $s x$ and $s' x'$ differ on one of their first $\min (|s|, |s'|)$ characters.  Therefore, $s \prec s'$ implies $s x \prec s' x'$.
\end{proof}

\begin{lemma}\label{lem:mapping}
For any suffix $x$ of $\mathtt{\#}\,T\,\mathtt{\$}^w$ with $|x| > w$, exactly one prefix $s$ of $x$ is in $S$.
\end{lemma}

\begin{proof}
Consider the substring $d$ stretching from the beginning of the last occurrence of an element of $E'$ that starts before or at the starting position of $x$, to the end of the first occurrence of an element of $E'$ that starts strictly after the starting position of $x$.  Regardless of whether $d$ starts with $\mathtt{\#}$ or another element of $E'$, it is prefixed by exactly one element of $E'$; similarly, it is suffixed by exactly one element of $E'$.  It follows that $d$ is an element of $D$.  Let $s$ be the prefix of $x$ that ends at the end of that occurrence of $d$ in $\mathtt{\#}\,T\,\mathtt{\$}^w$, so $|s| > w$ and is a suffix of an element of $D$ and thus $s \in S$.  By Lemma~\ref{lem:prefix-free}, no other prefix of $x$ is in $S$.
\end{proof}

Because of Lemma~\ref{lem:mapping}, we can define a function~$f$ mapping each suffix $x$ of $\mathtt{\#}\,T\,\mathtt{\$}^w$ with $|x| > w$ to the unique prefix $s$ of $x$ with $s \in S$.

\begin{lemma}
\label{lem:truncate}
Let $x$ and $x'$ be suffixes of $\mathtt{\#}\,T\,\mathtt{\$}^w$ with $|x|, |x'| > w$.  Then $f (x) \prec f (x')$ implies $x \prec x'$.
\end{lemma}

\begin{proof}
By the definition of $f$, $f (x)$ and $f (x')$ are prefixes of $x$ and $x'$ with $|f (x)|, |f (x')| > w$.  Therefore, $f (x) \prec f (x')$ implies $x \prec x'$ by Lemma~\ref{lem:comparable}.
\end{proof}

Define $T' [0..n] = T\,\mathtt{\$}$.  Let $g$ be the function that maps each suffix $y$ of $T'$ to the unique suffix $x$ of $\mathtt{\#}\,T\,\mathtt{\$}^w$ that starts with $y$, except that it maps $T' [n] = \mathtt{\$}$ to $\mathtt{\#}\,T\,\mathtt{\$}^w$.  Notice that $g (y)$ always has length greater than $w$, so it can be given as an argument to $f$.

\begin{lemma}
\label{lem:permutation}
The permutation that lexicographically sorts $T [0..n - 1]\,\mathtt{\$}^w, \ldots, T [n - 1]\,\mathtt{\$}^w, \mathtt{\#}\,T\,\mathtt{\$}^w$ also lexicographically sorts $T' [0..n], \ldots, T' [n - 1..n], T' [n]$.
\end{lemma}

\begin{proof}
Appending copies of $\mathtt{\$}$ to the suffixes of $T'$ does not change their relative order, and just as $\mathtt{\#}\,T\,\mathtt{\$}^w$ is the lexicographically smallest of $T [0..n - 1]\,\mathtt{\$}^w, \ldots, T [n - 1]\,\mathtt{\$}^w, \mathtt{\#}\,T\,\mathtt{\$}^w$, so $T' [n] = \mathtt{\$}$ is the lexicographically smallest of $T' [0..n], \ldots, T' [n - 1..n], T' [n]$.
\end{proof}

Let $\beta$ be the function that, for $i < n$, maps $T' [i]$ to the lexicographic rank of $f (g (T' [i + 1..n]))$ in $S$, and maps $T [n]$ to the lexicographic rank of $f (g (T')) = f (T\,\mathtt{\$}^w)$.  

\begin{lemma}
\label{lem:subsequence}
Suppose $\beta$ maps $k$ copies of $a$ to $s \in S$ and maps no other characters to $s$, and maps a total of $t$ characters to elements of $S$ lexicographically less than $s$.  Then the $(t + 1)$st through $(t + k)$th characters of the BWT of $T'$ are copies of $a$.
\end{lemma}

\begin{proof}
By Lemmas~\ref{lem:truncate} and~\ref{lem:permutation}, if $f (g (y)) \prec f (g (y'))$ then $y \prec y'$.  Therefore, $\beta$ partially sorts the characters in $T'$ into their order in the BWT of $T'$; equivalently, the characters' partial order according to $\beta$ can be extended to their total order in the BWT.  Since every total extension of $\beta$ puts those $k$ copies of $a$ in the $(t + 1)$st through $(t + k)$th positions, they appear there in the BWT.
\end{proof}

From $D$ and $P$, we can compute how often each element $s \in S$ is preceded by each distinct character $a$ in $\mathtt{\#}\,T\,\mathtt{\$}^w$ or, equivalently, how many copies of $a$ are mapped by $\beta$ to the lexicographic rank of $s$.  If an element $s \in S$ is a suffix of only one element $d \in D$ and a proper suffix of that --- which we can determine first from $D$ alone --- then $\beta$ maps only copies of the the preceding character of $d$ to the rank of $s$, and we can compute their positions in the BWT of $T'$.  If $s = d$ or a suffix of several elements of $D$, however, then $\beta$ can map several distinct characters to the rank of $s$.  To deal with these cases, we can also compute which elements of $D$ contain which characters mapped to the rank of $s$.  We will explain in a moment how we use this information.

For our example, $T = \mbox{\tt GATTACAT!GATACAT!GATTAGATA}$, we compute the information shown in Table~\ref{tab:example}.  To ease the comparison to the standard computation of the BWT of $T'\,\mathtt{\$}$, shown in Table~\ref{tab:bwt}, we write the characters mapped to each element $s \in S$ before $s$ itself.

By Lemma~\ref{lem:subsequence}, from the characters mapped to each rank by $\beta$ and the partial sums of frequencies with which $\beta$ maps characters to the ranks, we can compute the subsequence of the BWT of $T'$ that contains all the characters $\beta$ maps to elements of $S$, which are not complete elements of $D$ and to which only one distinct character is mapped.  We can also leave placeholders where appropriate for the characters $\beta$ maps to elements of $S$, which are complete elements of $D$ or to which more than one distinct character is mapped.  For our example, this subsequence is {\tt ATTTTTTCCGGGGAAA!\$!AAA\,-\,-\,TAA}.  Notice we do not need all the information in $P$ to compute this subsequence, only $D$ and the frequencies of its elements in $P$.

\begin{table}[t!]
\begin{center}
\caption{The information we compute for our example, $T = \mbox{\tt GATTACAT!GATACAT!GATTAGATA}$.  Each line shows the lexicographic rank $r$ of an element $s \in S$; the characters mapped to $r$ by $\beta$; $s$ itself; the elements of $D$ from which the mapped characters originate; the total frequency with which characters are mapped to $r$; and the preceding partial sum of the frequencies.}
\label{tab:example}
\tt
\begin{tabular}{cclccc}
&&&&\\[1ex]
           & {\rm mapped}     &                                &                                 &                 & {\rm preceding}\\
{\rm rank} & {\rm characters} & \multicolumn{1}{c}{\rm suffix} & \multicolumn{1}{c}{\rm sources} & {\rm frequency} & {\rm partial sum}\\
\hline\\[-2ex]
$0$  & A       & \#GATTAC  & $1$    & $1$ &  $0$\\
$1$  & T       & !GATAC    & $2$    & $1$ &  $1$\\
$2$  & T       & !GATTAG   & $3$    & $1$ &  $2$\\
$3$  & T       & A\$\$     & $5$    & $1$ &  $3$\\
$4$  & T       & ACAT!     & $4$    & $2$ &  $4$\\
$5$  & T       & AGATA\$\$ & $5$    & $1$ &  $6$\\
$6$  & C       & AT!       & $4$    & $2$ &  $7$\\
$7$  & G       & ATA\$\$   & $5$    & $1$ &  $9$\\
$8$  & G       & ATAC      & $2$    & $1$ & $10$\\
$9$  & G       & ATTAC     & $1$    & $1$ & $11$\\
$10$ & G       & ATTAG     & $3$    & $1$ & $12$\\
$11$ & A       & CAT\#     & $4$    & $2$ & $13$\\
$12$ & A       & GATA\$\$  & $5$    & $1$ & $15$\\
$13$ & !       & GATAC     & $2$    & $1$ & $16$\\
$14$ & \$      & GATTAC    & $1$    & $1$ & $17$\\
$15$ & !       & GATTAG    & $3$    & $1$ & $18$\\
$16$ & A       & T!GATAC   & $2$    & $1$ & $19$\\
$17$ & A       & T!GATTAG  & $3$    & $1$ & $20$\\
$18$ & A       & TA\$\$    & $5$    & $1$ & $21$\\
$19$ & T$,$\,A & TAC       & $1; 2$ & $2$ & $22$\\
$20$ & T       & TAG       & $3$    & $1$ & $24$\\
$21$ & A       & TTAC      & $1$    & $1$ & $25$\\
$22$ & A       & TTAG      & $3$    & $1$ & $26$
\end{tabular}
\end{center}
\end{table}

\begin{table}[t!]
\begin{center}
\caption{The BWT for $T' = \mbox{\tt GATTACAT!GATACAT!GATTAGATA\$}$.  Each line shows a position in the BWT; the character in that position; and the suffix immediately following that character in $T'$.}
\label{tab:bwt}
\tt
\begin{tabular}{ccl}
&&\\[1ex]
$i$ & $\BWT [i]$ & \multicolumn{1}{c}{\rm suffix}\\
\hline\\[-2ex]
 $0$ & A  & \$\\
 $1$ & T  & !GATACAT!GATTAGATA\$\\
 $2$ & T  & !GATTAGATA\$\\
 $3$ & T  & A\$\\
 $4$ & T  & ACAT!GATACAT!GATTAGATA\$\\
 $5$ & T  & ACAT!GATTAGATA\$\\
 $6$ & T  & AGATA\$\\
 $7$ & C  & AT!GATACAT!GATTAGATA\$\\
 $8$ & C  & AT!GATTAGATA\$\\
 $9$ & G  & ATA\$\\
$10$ & G  & ATACAT!GATTAGATA\$\\
$11$ & G  & ATTACAT!GATACAT!GATTAGATA\$\\
$12$ & G  & ATTAGATA\$\\
$13$ & A  & CAT!GATACAT!GATTAGATA\$\\
$14$ & A  & CAT!GATTAGATA\$\\
$15$ & A  & GATA\$\\
$16$ & !  & GATACAT!GATTAGATA\$\\
$17$ & \$ & GATTACAT!GATACAT!GATTAGATA\$\\
$18$ & !  & GATTAGATA\$\\
$19$ & A  & T!GATACAT!GATTAGATA\$\\
$20$ & A  & T!GATTAGATA\$\\
$21$ & A  & TA\$\\
$22$ & T  & TACAT!GATACAT!GATTAGATA\$\\
$23$ & A  & TACAT!GATTAGATA\$\\
$24$ & T  & TAGATA\$\\
$25$ & A  & TTACAT!GATACAT!GATTAGATA\$\\
$26$ & A  & TTAGATA\$\\
\end{tabular}
\end{center}
\end{table}

Suppose $s \in S$ is an entire element of $D$ or a suffix of several elements of $D$, and occurrences of $s$ are preceded by several distinct characters in $\mathtt{\#}\,T\,\mathtt{\$}^w$, so $\beta$ assigns $s$'s lexicographic rank in $S$ to several distinct characters.  To deal with such cases, we can sort the suffixes of the parse $P$ and apply the following lemma.

\begin{lemma}
\label{lem:parse}
Consider two suffixes $t$ and $t'$ of $\mathtt{\#}\,T\,\mathtt{\$}^w$ starting with occurrences of $s \in S$, and let $q$ and $q'$ be the suffixes of $P$ encoding the last $w$ characters of those occurrences of $s$ and the remainders of $t$ and $t'$.  If $t \prec t'$ then $q \prec q'$.
\end{lemma}

\begin{proof}
Since $s$ occurs at least twice in $\mathtt{\#}\,T\,\mathtt{\$}^w$, it cannot end with $\mathtt{\$}^w$ and thus cannot be a suffix of $\mathtt{\#}\,T\,\mathtt{\$}^w$.  Therefore, there is a first character on which $t$ and $t'$ differ.  Since the elements of $D$ are represented in the parse by their lexicographic ranks, that character forces $q \prec q'$.
\end{proof}

We consider the occurrences in $P$ of the elements of $D$ suffixed by $s$, and sort the characters preceding those occurrences of $s$ into the lexicographic order of the remaining suffixes of $P$ which, by Lemma~\ref{lem:parse}, is their order in the BWT of $T'$.  In our example, $\mathtt{TAC} \in S$ is preceded in $\mathtt{\#}\,T\,\mathtt{\$\$}$ by a {\tt T} when it occurs as a suffix of $\mathtt{\#GATTAC} \in D$, which has rank 0 in $D$, and by an {\tt A} when it occurs as a suffix of $\mathtt{T!GATAC} \in D$, which has rank 3 in $D$.  Since the suffix following 0 in $P = 0, 1, 3, 1, 4, 2$ is lexicographically smaller than the suffix following 3, that {\tt T} precedes that {\tt A} in the BWT.

Since we need only $D$ and the frequencies of its elements in $P$ to apply Lemma~\ref{lem:subsequence} to build and store the subsequence of the BWT of $T'$ that contains all the characters $\beta$ maps to elements of $S$, to which only one distinct character is mapped, this takes space proportional to the total length of the elements of $D$.  We can then apply Lemma~\ref{lem:parse} to build the subsequence of missing characters in the order they appear in the BWT.  Although this subsequence of missing characters could take more space than $D$ and $P$ combined, as we generate them we can interleave them with the first subsequence and output them, thus still using workspace proportional to the total length of $P$ and the elements of $D$ and only one pass over the space used to store the BWT.

Notice, we can build the first subsequence from $D$ and the frequencies of its elements in $P$; store it in external memory; and make a pass over it while we generate the second one from $D$ and $P$, inserting the missing characters in the appropriate places.  This way we use two passes over the space used to store the BWT, but we may use significantly less workspace.

Summarizing, assuming we can recognize the strings in $E$ quickly, we can quickly compute $D$ and $P$ with one scan over $T$. From $D$ and $P$, with Lemmas~\ref{lem:subsequence} and~\ref{lem:parse}, we can compute the BWT of $T' = T\,\mathtt{\$}$ by sorting the suffixes of the elements of $D$ and the suffixes of $P$.  Since there are linear-time and linear-space algorithms for sorting suffixes when working in internal memory, this implies our main theoretical result:

\begin{theorem}
\label{thm:pfp}
We can compute the BWT of $T\,\mathtt{\$}$ from $D$ and $P$ using workspace proportional to sum of the total length of $P$ and the elements of $D$, and $O (n)$ time when we can work in internal memory.
\end{theorem}

The significance of the above theorem is that if the text $T$ contains many repetitions the dictionary of distinct phrases $D$ will be relatively small and, if the dictionary words are sufficiently long, also the parse $P$ will be much smaller than $T$. Thus, even if $T$ is very large, if $D$ and $P$ fit into internal memory then using Theorem~\ref{thm:pfp} we can efficiently build the BWT for $T$ doing the critical computations in RAM after a single sequential scanning of $T$ during the parsing phase. 



\section{Prefix free parsing in practice}
\label{sec:practice}

We have implemented our BWT construction based on prefix free parsing and applied it to collections of repetitive documents and genomic databases. Our purpose is to test our conjectures that 1) with a good choice of the parsing strategy the total length of the phrases in the dictionary and the number of phrases in the parse will both be small in comparison to the uncompressed size of the collection, and 2) computing the BWT from the dictionary and the parse leads to an overall speed-up and reduction in memory usage. In this section we describe our implementation and then report our experimental results.
 

%


\subsection{Implementation}
\label{subsec:implementation}

Given a window size~$w$, we select a prime $p$ and we define the set~$E$ described in Section~\ref{sec:theory}, as the set of length-$w$ strings such that their Karp-Rabin fingerprint modulo $p$ is zero. 
With this choice our parsing algorithm works as follows. We slide a window of length $w$ over the text, keeping track of the Karp-Rabin hash of the window; we also keep track of the hash of the entire prefix of the current phrase that we have processed so far.  Whenever the hash of the window is 0 modulo $p$, we terminate the current phrase at the end of the window and start the next one at the beginning of the window.  We prepend a NUL character to the first phrase and append $w$ copies of NUL to the last phrase.  If the text ends with $w$ characters whose hash is 0 modulo $p$, then we take those $w$ character to be the beginning of the last phrase and append to them $w$ copies of NUL.  We note that we prepend and append copies of the same NUL character; although using different characters simplifies the proofs in Section~\ref{sec:theory}, it is not essential in practice.

We keep track of the set of hashes of the distinct phrases in the dictionary so far, as well as the phrases' frequencies.  Whenever we terminate a phrase, we check if its hash is in that set.  If not, we add the phrase to the dictionary and its hash to the set, and set its frequency to 1; if so, we compare the current phrase to the one in the dictionary with the same hash to ensure they are equal, then increment its frequency.  (Using a 64-bit hash the probability of there being a collision is very low, so we have not implemented a recovery mechanism if one occurs.)  In both cases, we write the hash to disk.

When the parsing is complete, we have generated the dictionary $D$ and the parsing $P = p_1, p_2,\ldots, p_z$, where each phrase $p_i$ is represented by its hash.  Next, we sort the dictionary and make a pass over $P$ to substitute the phrases' lexicographic ranks for their hashes. This gives us the final parse, still on disk, with each entry stored as a 4-byte integer.  We write the dictionary to disk phrase by phrase in lexicographic order with a special end-of-phrase terminator at the end of each phrase. In a separate file we store the frequency of each phrase in as a 4-byte integer.  Using four bytes for each integer does not give us the best compression possible, but it makes it easy to process the frequency and parse files later. Finally, we write to a separate file the array $W$ of length $|P|$ such that $W[j]$ is the character of $p_j$ in position $w+1$ from the end (recall each phrase has length greater than $w$). These characters will be used to handle the elements of $S$ that are also elements of $D$.

Next, we compute the BWT of the parsing $P$, with each phrase represented by its 4-byte lexicographic rank in $D$. The computation is done using the SACA-K suffix array construction algorithm~\cite{tois/Nong13} which, among the linear time algorithms, is the one using the smallest workspace and is particularly suitable for input over large alphabets. Instead of storing $BWT(P) = b_1, b_2,\ldots, b_z$, we save the same information in a format more suitable for the next phase. We consider the dictionary phrases in lexicographic order, and, for each phrase $d_i$, we write the list of BWT positions where $d_i$ appears. We call this the inverted list for phrase $d_i$. Since the size of the inverted list of each phrase is equal to its frequency, which is available separately, we write to file the plain concatenation of the inverted lists using again four bytes per entry, for a total of $4|P|$ bytes. In this phase we also permute the elements of $W$ so that now $W[j]$ is the character coming from the phrase that precedes $b_j$ in the parsing, i.e. $P[SA[j]-2]$.

In the final phase of the algorithm we compute the BWT of the input~$T$. We deviate slightly from the description in Section~\ref{sec:theory} in that instead of lexicographically sorting the suffixes in $D$ larger than $w$ we sort all of them and later ignore those which are of length $\leq w$. The sorting is done applying the gSACAK algorithm~\cite{tcs/LouzaGT17} which computes the SA and LCP array for the set of dictionary phrases. We then proceed as in Section~\ref{sec:theory}. If during the scanning of the sorted set $S$ we meet $s$ which is a proper suffix of several elements of $D$ we use a heap to merge their respective inverted lists writing a character to the final BWT file every time we pop a position from the heap. If we meet $s$ which coincides with a dictionary phrase $d$ we write the characters retrieved from $W$ from the positions obtained from $d$'s inverted list. 

It turns out that the the most expensive phases of the algorithm are the first, where we compute the parsing of $T$, and the last, where we compute $BWT(T)$ from the SA of $D$ and the inverted lists for $D$'s phrases. Fortunately, both phases can be sped-up using multiple threads in parallel. To parallelize the first phase we split the input into equal size chunks, and we assign each chunk to a different thread. Using this simple approach, we obtained a speed-up of a factor 2 using four threads, but additional threads do not yield substantial improvements. There are two likely reasons for that: 1) during the parsing all threads need to update the same dictionary, and 2) each thread has to write to disk its portion of the parsing and, unless the system has multiple disks, disk access becomes a bottleneck. To parallelize the computation of the final $BWT(T)$ we use a different approach. The main thread scans the suffix array of the dictionary and as soon as it finds a range of equal suffixes it passes such range to an helper thread that computes and writes to disk the corresponding portion of $BWT(T)$. Again, we were able to reduce the running time of this phase by factor 2 using four threads. In the next section we only report the running times for the single thread algorithm since we are still working to improve our multi-thread version.

\subsection{Experiments}
\label{subsec:experiments}
In this section, the parsing and BWT computation are experimentally evaluated. All experiments were run on a server with Intel(R) Xeon(R) CPU E5-2640 v4 @ 2.40GHz and $756$ gigabytes of RAM.

Table~\ref{tab:pizzachili} shows the sizes of the dictionaries and parses for several files from the Pizza \& Chili repetitive corpus~\cite{repcorpus}, with three settings of the parameters $w$ and $p$.  We note that {\tt cere} contains long runs of {\tt N}s and {\tt world\_leaders} contains long runs of periods, which can either cause many phrases, when the hash of $w$ copies of those characters is 0 modulo $p$, or a single long phrase otherwise; we also display the sizes of the dictionaries and parses for those files with all {\tt N}s and periods removed.  The dictionaries and parses occupy between 5 and 31 percent of the space of the uncompressed files.

\begin{table}
\begin{center}
\caption{The dictionary and parse sizes for several files from the Pizza \& Chili repetitive corpus, with three settings of the parameters $w$ and $p$.  All sizes are reported in megabytes; percentages are the sums of the sizes of the dictionaries and parses, divided by the sizes of the uncompressed files.}
\label{tab:pizzachili}
\begin{tabular}{r|r|rrr|rrr|rrr}
\multicolumn{1}{c}{} & \multicolumn{1}{c}{} & \multicolumn{3}{c}{$w = 6, p = 20$} & \multicolumn{3}{c}{$w = 8, p = 50$} & \multicolumn{3}{c}{$w = 10, p = 100$}\\
file & size & dict. & parse & \% & dict. & parse & \% & dict. & parse & \% \\
\hline
{\tt cere}
     &  440 &   61 &    77 & 31 &       43 &      159 & 46       & {\bf 89} & {\bf 17} & {\bf 24} \\
{\tt cere\_no\_Ns}
     &  409 &   33 &    77 & 27 & {\bf 43} & {\bf 33} & {\bf 18} &       60 &       17 &       19 \\
{\tt dna.001.1}
     &  100 &    8 &    20 & 27 & {\bf 13} & {\bf  9} & {\bf 21} &       21 &        4 &       25 \\
{\tt einstein.en.txt}
     &  446 &    2 &    87 & 20 &        3 &       39 &        9 & {\bf  4} & {\bf 17} & {\bf  5} \\
{\tt influenza}
     &  148 &   16 &    28 & 30 & {\bf 32} & {\bf 12} & {\bf 29} &       49 &        6 &       37 \\
{\tt kernel}
     &  247 &   14 &    52 & 26 &       14 &       20 &       13 & {\bf 15} & {\bf 10} & {\bf 10} \\
{\tt world\_leaders}
     &   45 &    5 &     5 & 21 & {\bf  8} & {\bf  2} & {\bf 21} &       11 &        1 &       26 \\
{\tt world\_leaders\_no\_dots}
     &   23 &    4 &     5 & 34 & {\bf  6} & {\bf  2} & {\bf 31} &        7 &        1 &       33
\end{tabular}
\end{center}
\end{table}
Table~\ref{tab:salmonella} shows the sizes of the dictionaries and parses for prefixes of a database of Salmonella genomes~\cite{STBASBM17}.  The dictionaries and parses occupy between 14 and 44 percent of the space of the uncompressed files, with the compression improving as the number of genomes increases.  In particular, the dataset of ten thousand genomes takes nearly 50 GB uncompressed, but with $w = 10$ and $p = 100$ the dictionary and parse take only about 7 GB together, so they would still fit in the RAM of a commodity machine.  This seems promising, and we hope the compression is even better for larger genomic databases.

Table~\ref{tab:pfbwt-genome} shows the runtime and peak memory usage for computing the BWT from the parsing for the database of Salmonella genomes. As a baseline for comparison, {\tt simplebwt} computes the BWT by first computing the Suffix Array using algorithm SACA-K~\cite{tois/Nong13} which is the same tool used internally by our algorithm since it is fast and uses $O(1)$ workspace. As shown in Table~\ref{tab:pfbwt-genome}, the peak memory usage of {\tt simplebwt} is reduced by a factor of $4$ to $10$ by computing the BWT from the parsing; furthermore, the total runtime is competitive with {\tt simplebwt}. In some instances, for example the databases of $5000$, $10,000$ genomes, computing the BWT from the parsing achieved significant runtime reduction over {\tt simplebwt}; with $w = 10$, $p = 100$ on these instances the runtime reduction is more than factors of $2$ and $4$ respectively. For our BWT computations, the peak memory usage with $w = 6$, $p = 20$ stays within a factor of roughly 2 of the original file size and is smaller than the original file size on the larger databases of $1000$ genomes.


Qualitatively similar results on files from the Pizza \& Chili corpus are shown
in Table \ref{tab:pfbwt-pizza}.

\subsubsection{On the choice of the parameter $w$ and $p$}
Finally, Fig. \ref{fig:wp-exps} shows the peak memory usage and runtime for computing the BWT
on a collection of $1000$ Salmonella genomes of size $2.7$ gigabytes, for all pairs of parameter
choices $(w,p)$, where $w \in \{6,8,10\}$ and $p \in \{ 50, 100, 200, 400, 800 \}$.
As shown in Fig. \ref{fig:wp-exp-mem}, the choice $(w,p)=(10,50)$ results in the smallest
peak memory usage, while Fig. \ref{fig:wp-exp-time} shows that
$(w,p)=(10,200)$ results in the fastest runtime. In general, for either 
runtime or peak memory usage, the optimal choice of $(w,p)$ differs 
and depends on the input. However, notice that all pairs $(w,p)$ tested
here required less than $1.1$ times the input size of memory and the slowest runtime
was less than twice the fastest. 
\begin{table}
\begin{center}
\caption{The dictionary and parse sizes for prefixes of a database of Salmonella genomes, with three settings of the parameters $w$ and $p$.  Again, all sizes are reported in megabytes; percentages are the sums of the sizes of the dictionaries and parses, divided by the sizes of the uncompressed files.}
\label{tab:salmonella}
\begin{tabular}{r|r|rrr|rrr|rrr}
\multicolumn{1}{c}{number of} & \multicolumn{1}{c}{} & \multicolumn{3}{c}{$w = 6, p = 20$} & \multicolumn{3}{c}{$w = 8, p = 50$} & \multicolumn{3}{c}{$w = 10, p = 100$}\\
genomes & size & dict. & parse & \% & dict. & parse & \% & dict. & parse & \% \\
\hline
50 & 249 & 68 & 43 & 44 & {\bf 77} & {\bf 20} & {\bf 39} & 91 & 10 & 40\\
100 & 485 & 83 & 85 & 35 & {\bf 99} & {\bf 39} & {\bf 28} & 122 & 19 & 29\\
500 & 2436 & 273 & 424 & 29 & 314 & 194 & 21 & {\bf 377} & {\bf 96} & {\bf 19}\\
1000 & 4861 & 475 & 847 & 27 & 541 & 388 & 19 & {\bf 643} & {\bf 192} & {\bf 17}\\
5000 & 24936 & 2663 & 4334 & 28 & 2915 & 1987 & 20 & {\bf 3196} & {\bf 985} & {\bf 17}\\
10000 & 49420 & 4190 & 8611 & 26 & 4652 & 3939 & 17 & {\bf 5176} & {\bf 1955} & {\bf 14}
\end{tabular}
\end{center}

\bigskip

\begin{center}
\caption{Time (seconds) and peak memory consumption (megabytes) of BWT calculations for prefixes of a database of Salmonella genomes, for three settings of the parameters $w$ and $p$ and for the comparison method {\tt simplebwt}.}
\label{tab:pfbwt-genome}
\begin{tabular}{r|rr|rr|rr|rr}
\multicolumn{1}{c}{number of} & \multicolumn{2}{c}{$w = 6, p = 20$} & \multicolumn{2}{c}{$w = 8, p = 50$} & \multicolumn{2}{c}{$w = 10, p = 100$} & \multicolumn{2}{c}{\tt simplebwt}\\
genomes & time & peak  & time & peak & time & peak & time & peak \\
\hline
50    & 71     & {\bf 545}   & 63  &  642    & 65  & 782     & { \bf 53} & 2247   \\
100   & 118    & { \bf 709}  & { \bf 100}    & 837  & 102    & 1059      & 103 & 4368  \\
500   & 570    & {\bf 2519}  & 443 & 2742    & { \bf 402}    & 3304      & 565  & 21923 \\
1000  & 1155   & {\bf 4517}  & 876  & 4789   & { \bf 776}    & 5659      & 1377 & 43751 \\
5000  & 7412   & {\bf 42067} & 5436 & 46040  & { \bf 4808}   & 51848     & 11600 & 224423 \\
10000 & 19152  & {\bf 68434} & 12298 & 74500 & { \bf 10218}  & 84467     & 43657 & 444780 
\end{tabular}
\end{center}

\bigskip

\begin{center}
\caption{Time (seconds) and peak memory consumption (megabytes) of BWT calculations on various files from the Pizza \& Chili repetitive corpus, for three settings of the parameters $w$ and $p$ and for the comparison method {\tt simplebwt}.}
\label{tab:pfbwt-pizza}
\begin{tabular}{r|rr|rr|rr|rr}
\multicolumn{1}{c}{} & \multicolumn{2}{c}{$w = 6, p = 20$} & \multicolumn{2}{c}{$w = 8, p = 50$} & \multicolumn{2}{c}{$w = 10, p = 100$} & \multicolumn{2}{c}{\tt simplebwt}\\
file & time & peak  & time & peak & time & peak & time & peak \\
\hline
{ \tt cere} & 90 & 603 & 79 & {\bf 559} & {\bf 74} & 801 & 90 & 3962 \\
{ \tt einstein.en.txt} & 53 & 196 & 40 & 88  & {\bf 35} & {\bf 53} & 97 & 4016   \\
{ \tt influenza}       & {\bf 27} & {\bf 166} & 27 & 284 & 33 & 435 & 30 & 1331  \\
{ \tt kernel}          & 43 & 170 & 29 & {\bf 143} & {\bf 25} & 144 & 50 & 2216  \\
{ \tt world\_leaders}  & 7  & {\bf 50}  & 7  & 74  & {\bf 7}  & 98  & 7  & 405   \\
\end{tabular}
\end{center}
\end{table}

\begin{figure}
\centering
\subfloat[Peak memory (MB) vs. $(w,p)$]{ \label{fig:wp-exp-mem} \includegraphics[width=.45\textwidth]{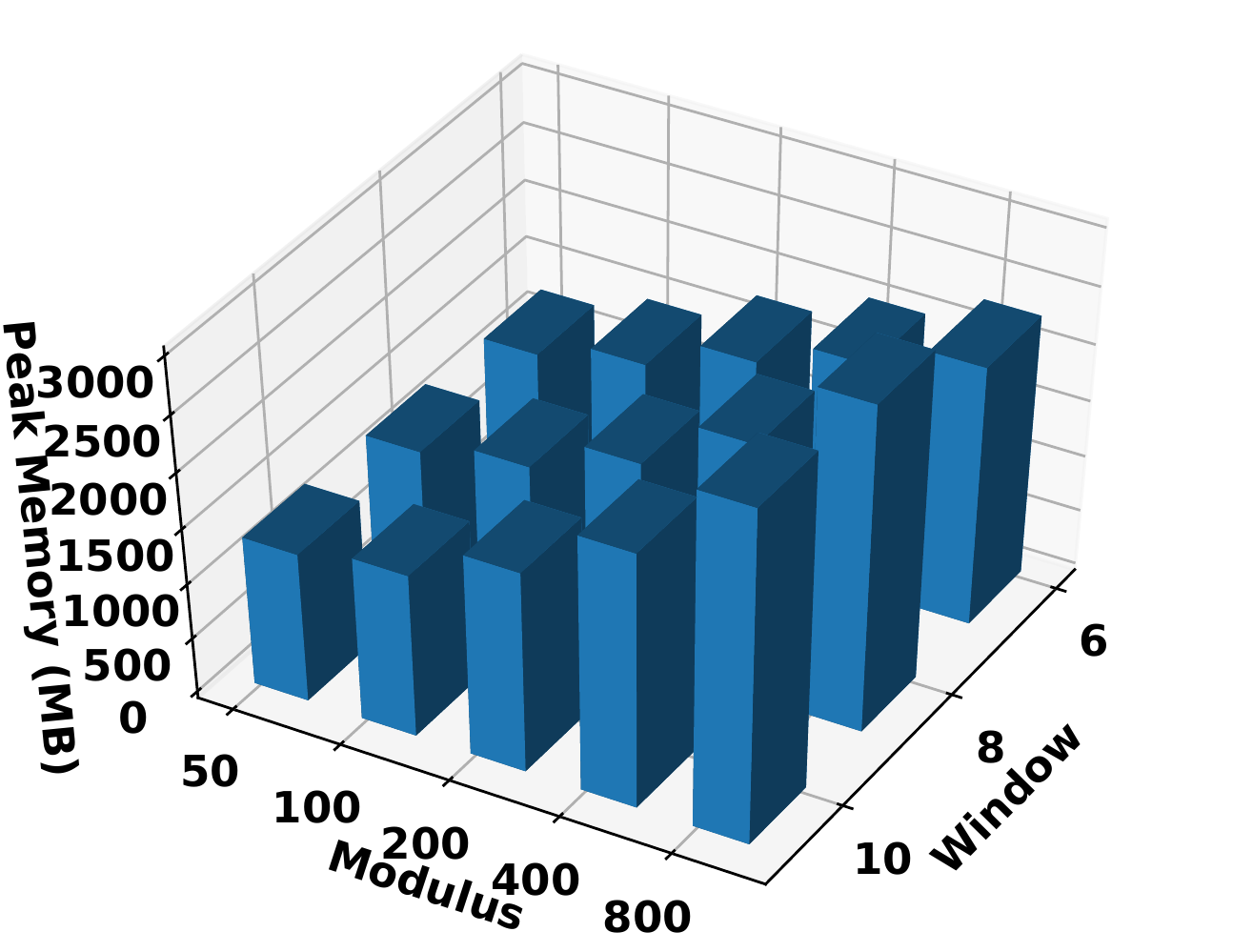} }
\subfloat[Time (s) vs. $(w,p)$]{ \label{fig:wp-exp-time} \includegraphics[width=.45\textwidth]{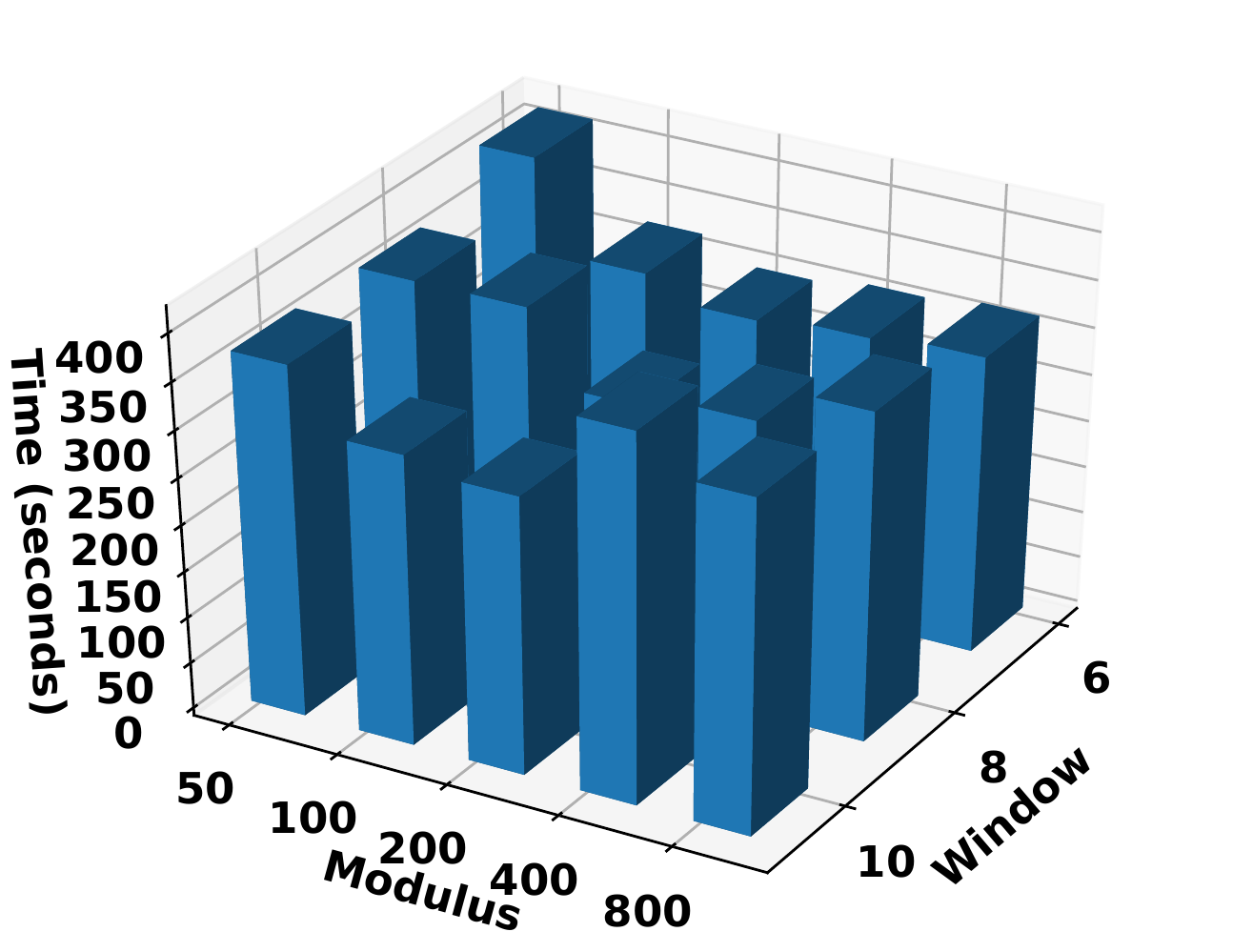} }
\caption{Results versus various choices of parameters $(w,p)$ on a collection
of 1000 Salmonella genomes (2.7 GB).}
\label{fig:wp-exps}
\end{figure}



\section{Indexing}
\label{sec:indexing}

The BWT is widely used as part of the FM index~\cite{FM05}, which is the heart
of popular DNA sequencing read aligners such as Bowtie~\cite{LTPS09,LS12}, BWA~\cite{LD10} and
SOAP 2~\cite{LYLLYKW09}. In these tools, rank support is added to 
the BWT using sampled arrays of precalculated ranks.  Similarly, locate support
is added using a sampled suffix array (SA).
Until recently, SA samples for massive, highly repetitive
datasets were much larger than the BWT, slow to calculate, or both.  Gagie,
Navarro, and Prezza~\cite{GNP18} have shown that only the SA values
at the ends of runs in the BWT need to be stored.  We are currently studying
how to build this sample during the process of computing the BWT from the
dictionary and the parse.
We show that prefix-free parsing can be incorporated into the construction of a counting-only run-length
FM index (RLFM) and we measure the time and space efficiency of the RLFM
construction and its ``count'' query in a DNA sequencing context
using data from the 1000 Genomes Project.
We compare the performance of the RLFM based methods to the popular Bowtie~\cite{LTPS09} read aligner.

\subsection{Implementation}

Constructing the counting-only RLFM requires three steps: building the BWT
from the text, generating the $F$ array, and run-length encoding
the BWT. We use prefix-free parsing to build the BWT. The $F$ array is easily built in a single pass over
the text. Run-length encoding is performed using the implementation by
Gagie, et.  al~\cite{GNP18}, which draws upon data structures implemented in
the Succinct Data Structure Library (SDSL) \cite{gbmp2014sea}; the concatenated
run-heads of the BWT are stored in a Huffman shaped wavelet tree, and auxiliary
bit-vectors are used to refer to the positions of the runs within the BWT.
During index construction, all characters that are not A, C, G, T, or N are
ignored.  

Typically, the BWT is built from a full SA, and thus
a sample could be built by simply retaining parts of the initial SA.  However,
prefix-free parsing takes a different approach, so to build a SA sample the
method would either need to be modified directly or a SA sample would have to
be generated \emph{post-hoc}. In the latter case, we can build a SA
sample solely from the BWT by ``back-stepping'' through the BWT with LF
mappings, and storing samples only at run-starts and run-ends. The main caveats
to this method are that an LF mapping would have to be done for every character
in the text, and that the entire BWT would need to be in memory in some form.
These drawbacks would be especially noticeable for large collections. 
 For now, we focus only on counting
support, and so SA samples are excluded from these
experiments except where otherwise noted.

\subsection{Experiments}

The indexes were built using data from the 1000 Genomes Project (1KG)
~\cite{consortium_global_2015}. Distinct versions of human chromosome 19 (``chr19'')
were created by using the \texttt{bcftools consensus} \cite{bcftools} tool to combine the
chr19 sequence from the GRCh37 assembly with a single haplotype
(maternal or paternal) from an individual in the 1KG project.
Chr19 is 58 million DNA bases long and makes up
1.9\% of the overall human genome sequence.
In all, we built 10 sets of chr19s, containing 1, 2, 10, 30, 50, 100, 250, 500, and 1000
distinct versions, respectively. This allows us to measure running time, peak memory
footprint and index size as a function of the collection size.
Index-building and counting experiments were run on a server with
Intel(R) Xeon(R) CPU E5-2650 v4 @ 2.20GHz and $512$ gigabytes of RAM. 

\subsubsection{Index building}

We compared our computational efficiency to that of Bowtie~\cite{LTPS09} v1.2.2, using the
\texttt{bowtie-build} command to build Bowtie indexes for each collection.
We could not measure beyond the 250 distinct versions as Bowtie exceeded available memory.
Peak memory was measured using the Unix \texttt{time -v} command, taking
the value from its ``Maximum resident set size (kbytes)'' field. Timings were
calculated and output by the programs themselves. The peak memory footprint
and running time for RLFM index building for each collection are plotted in
Figure \ref{fig:chr19-exps}.

\begin{figure}
\centering
\subfloat[Peak memory usage]{
\label{fig:chr19-constr-mem} \includegraphics[width=.45\textwidth]{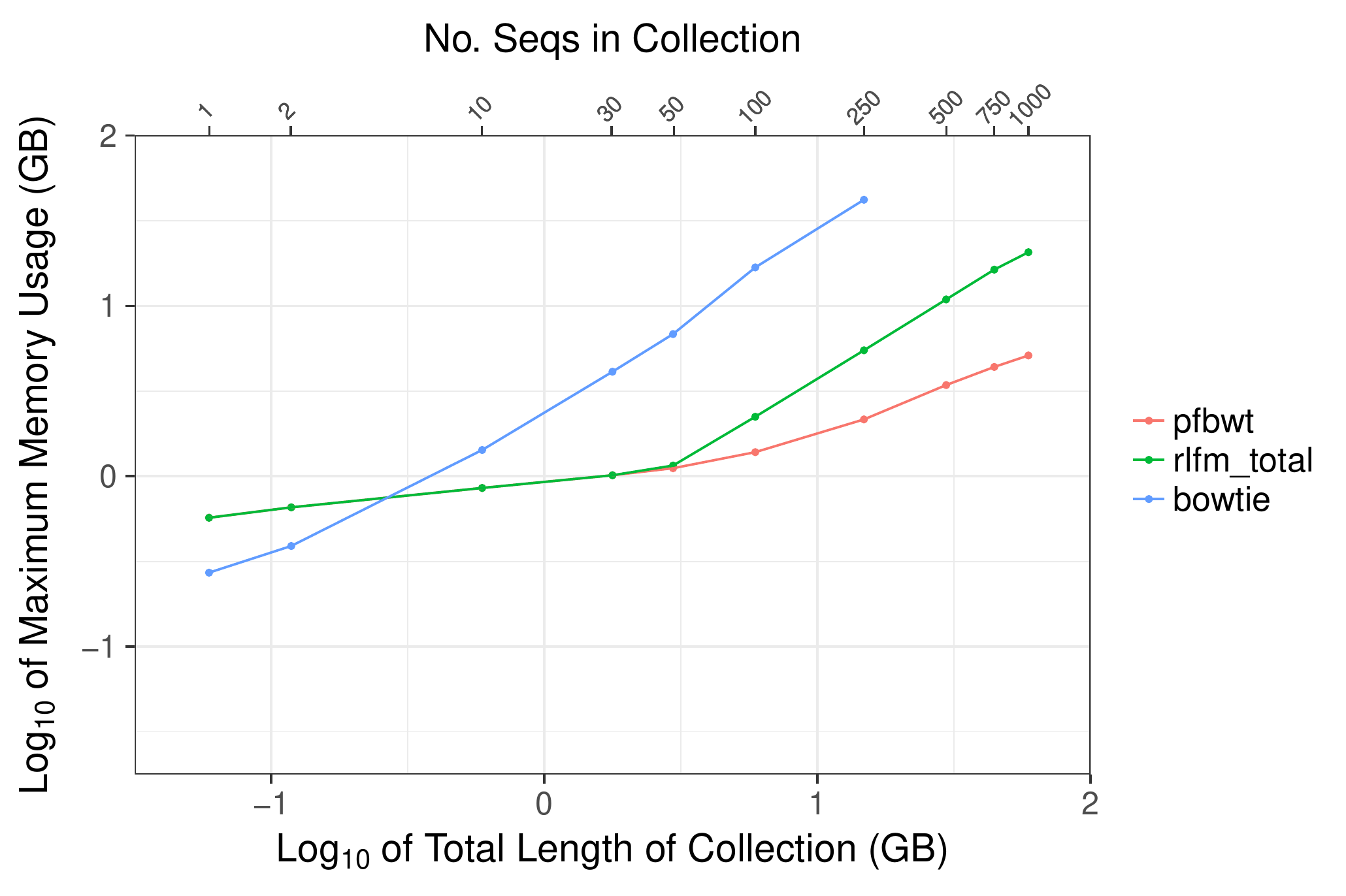} }
\subfloat[Running time]{ \label{fig:chr19-constr-time} \includegraphics[width=.45\textwidth]{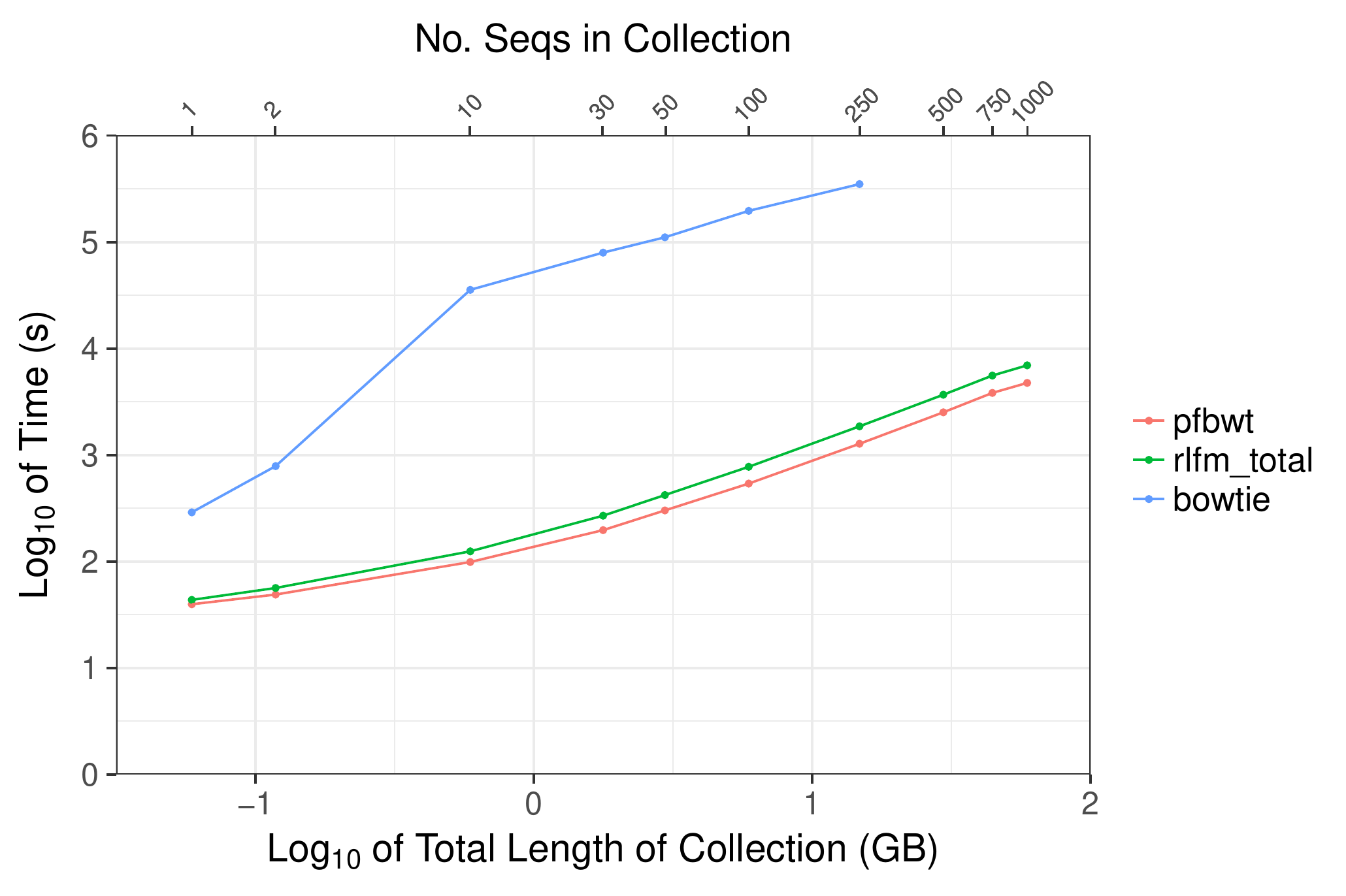} }
\caption{RLFM indexing efficiency for successively larger collections of genetically distinct human chr19s.
Results for the prefix-free parsing step (``pfbwt'') are shown
alongside the overall RLFM index-building (``rlfm\_total'')
and Bowtie (``bowtie'') results.}
\label{fig:chr19-exps}
\end{figure}

Compared to Bowtie, the resources required for RLFM index-building grew much
more slowly.  For example, the RLFM required about 20 GB to build an index for
1,000 chr19 sequences, whereas Bowtie required twice that amount to build an
index for just 250 sequences.
For data points up to 50 sequences in Figure \ref{fig:chr19-exps}a, the ``pfbwt''
and ``rlfm\_total'' points coincided, indicating that prefix-free parsing drove peak
memory footprint for the overall index-building process.  After 50 sequences,
however, ``pfbwt'' fell below ``rlfm\_total'' and accounted for a diminishing
fraction of the footprint as the collection grew.
Similarly, prefix-free parsing accounted for a diminishing fraction of total index-building time
as the sequence collection grew (Figure \ref{fig:chr19-exps}b).
These trends illustrate the advantage of prefix-free parsing when collections
are large and repetitive.


We can extrapolate the time and memory required to index many whole human genomes.
Considering chr19 accounts for 1.9\% of the human genome sequence,
and assuming that indexing 1,000 whole human-genome haplotypes will therefore require about 52.6 times
the time and memory as indexing 1,000 chr19s, we extrapolate that
indexing 1,000 human haplotypes would incur a peak memory footprint of about 1.01 TB
and require about 102 hours to complete.
Of course, the latter figure can be improved with parallelization.

We also measured that the index produced for the 1,000 chr19s
took about 131MB of disk space.
Thus, we can extrapolate that the index for 1,000 
human haplotypes would take about 6.73 GB.
While this figure makes us optimistic about future scaling, it is not directly comparable
to the size of a Bowtie genome index since it excludes the SA samples
needed for ``locate'' queries.


\subsubsection{Count query time}

We measured how the speed of the RLFM ``count'' operation scales with the size
of the sequence collection. Given a string pattern, ``count''
applies the LF mapping repeatedly to obtain the range of
SA positions matching the pattern.
It returns the size of this range.

We measured average ``count'' time by conducting a simple simulation of DNA-sequencing-like data.
First we took the first chr19 version and extracted and saved 100,000 random substrings of length 100.
That chr19 was included in all the successive collections,
so these substrings are all guaranteed to occur at least once
regardless of which collection we are querying.

We then queried each of the collections with the 100,000 substrings and divided the
running time by 100,000 to obtain the average ``count'' query time.
Query time increases as the collection grows (Figure \ref{fig:rlfm-count-exp})
but does so slowly, increasing from 750
microseconds for 1 sequence to 933 microseconds for 1,000 sequences.

\begin{figure}
\centering
\label{fig:rlfm-count} \includegraphics[width=.45\textwidth]{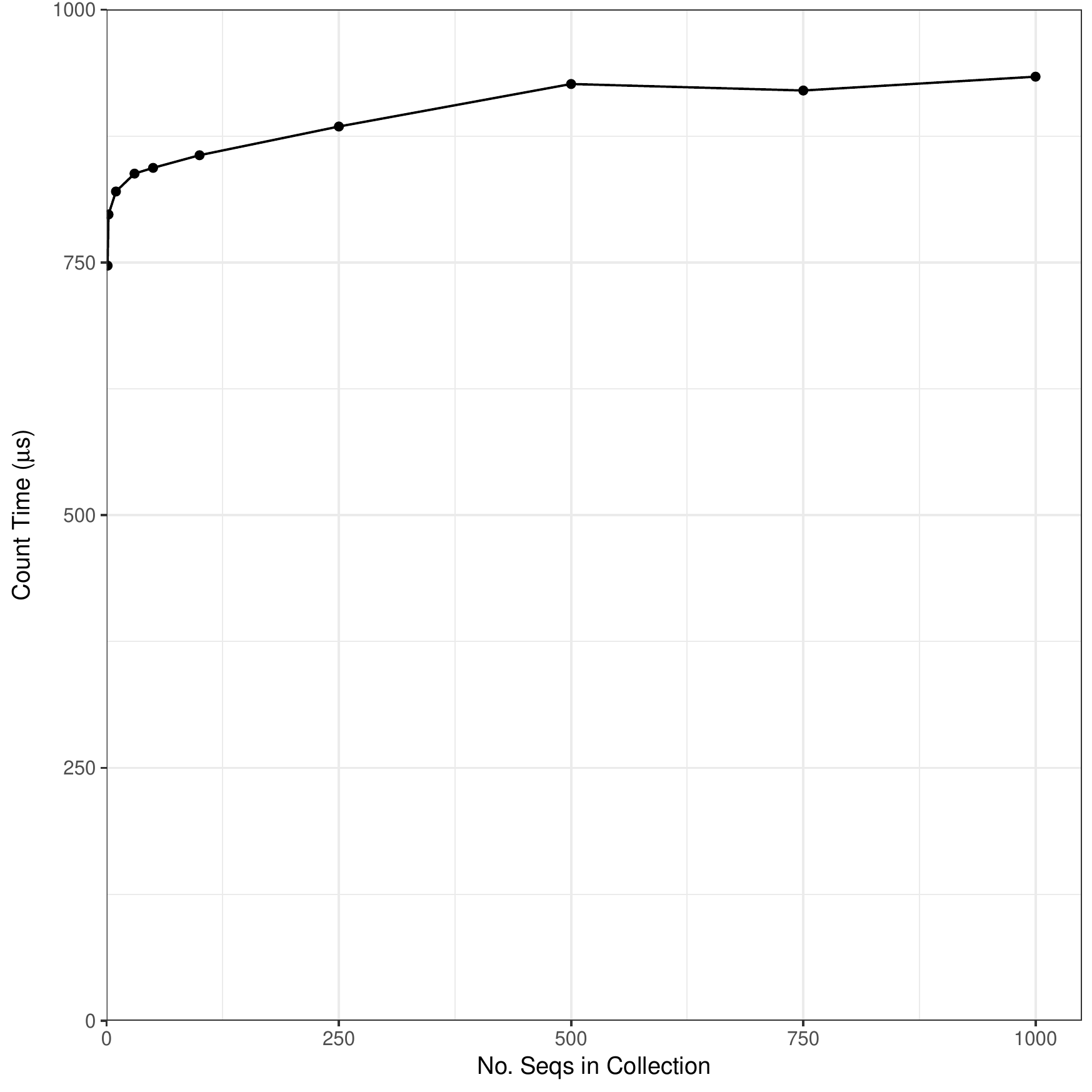}
\caption{RLFM average ``count'' query time for successively larger collections of genetically distinct human chr19s.  }
\label{fig:rlfm-count-exp}
\end{figure}

\subsubsection{Including the SA sample}

Though no SA sample was built for the experiments described so far,
such a sample is needed for ``locate'' queries that return the text offset
corresponding to a BWT element.
A SA sample can be obtained using the ``back-stepping'' method
described above.
We implemented a preliminary version of this method to examine whether
prefix-free parsing is a bottleneck in that scenario. 
Here, index building consists of three steps: (1) building the BWT
using prefix-free parsing (``pfbwt''), 
(2) back-stepping to create the SA sample and auxiliary structures (``bwtscan''),
and (3) run-length encoding the BWT (``rle'').
We built RLFM indexes for the same chr19
collections as above, measuring running time and peak memory footprint
for each of the three index-building step independently (Figure \ref{fig:ssa-exps}).

\begin{figure}
\centering
\subfloat[Peak memory usage]{
\label{fig:ssa-constr-mem} \includegraphics[width=.45\textwidth]{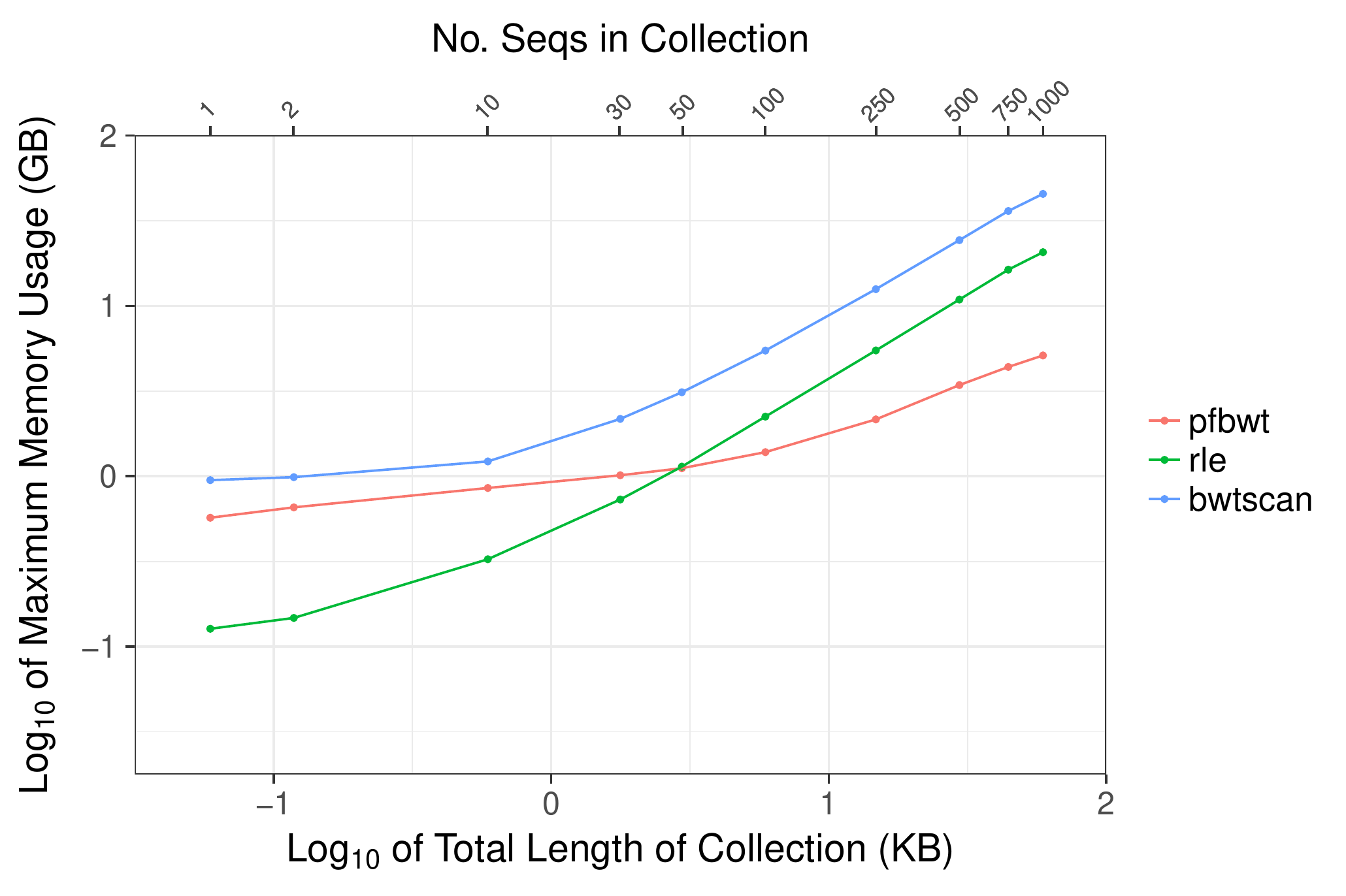} }
\subfloat[Running time]{ \label{fig:ssa-constr-time} \includegraphics[width=.45\textwidth]{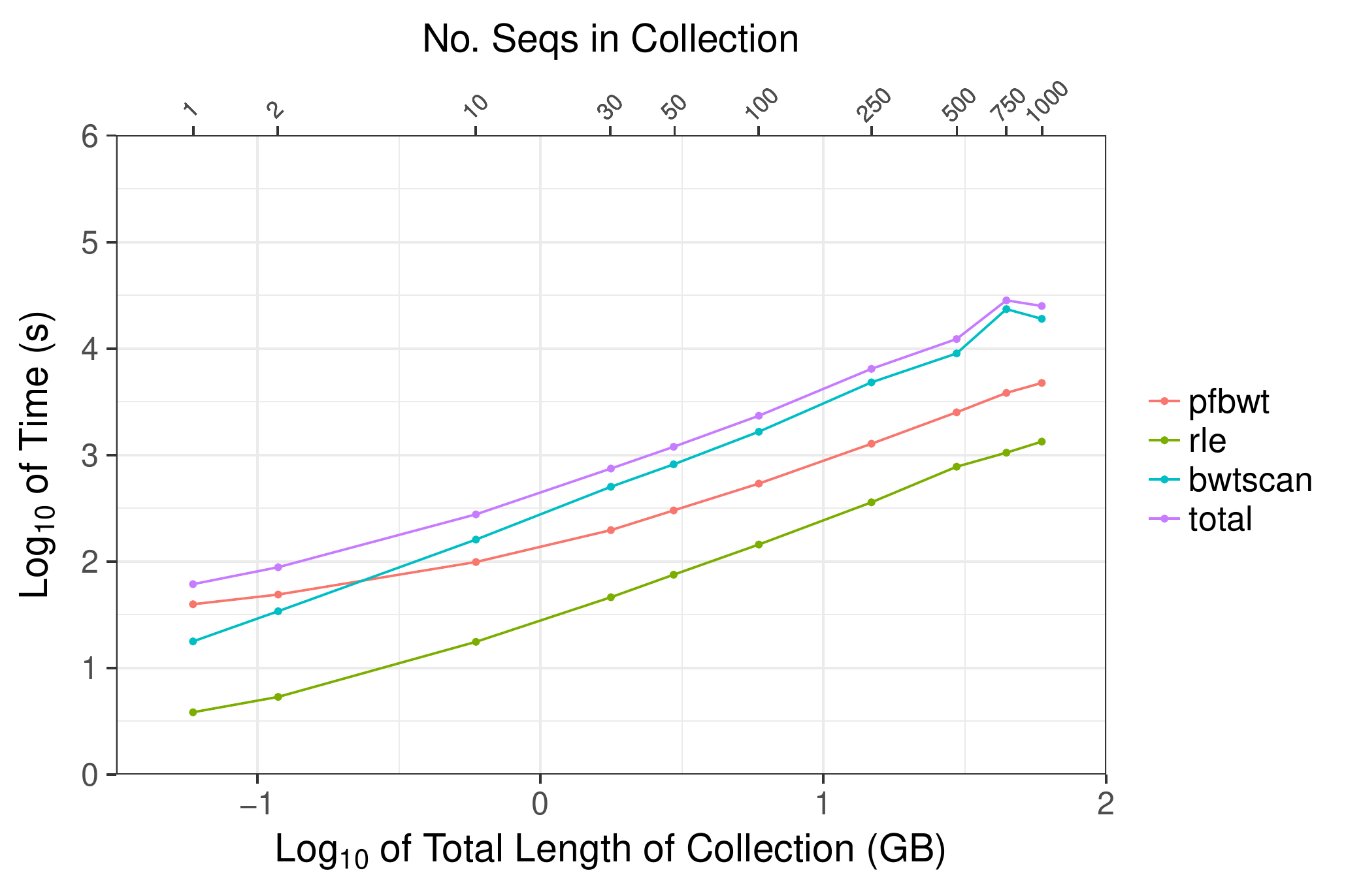} }
\caption{Computational efficiency of the three stages of index building
    when SA sampling is included.  
   Results are shown for the
    prefix-free parsing (``pfbwt''), back-stepping (``bwtscan'') and run-length
    encoding (``rle'') steps. ``total'' is the
    sum of the three steps.}
\label{fig:ssa-exps}
\end{figure}

The ``bwtscan'' step consistently drives peak memory footprint, since
it stores the entire BWT in memory as a Huffman shaped wavelet tree
\cite{gbmp2014sea}.
The ``pfbwt'' step has a substantially smaller footprint and used
the least memory of all the steps for collections larger than 50 sequences.
``pfbwt'' is the slowest step for small collections,
but ``bwtscan'' becomes the slowest for 10 or more sequences.
We conclude that as the collection of
sequences grows, the prefix-free parsing method
contributes proportionally less to peak memory footprint and running time,
and presents no bottlenecks for large collections.



\section{Conclusions}

We have described how prefix-free parsing can be used as preprocessing step to enable compression-aware computation of BWTs of large genomic databases.  Our results demonstrate that the dictionaries and parses are often significantly smaller than the original input, and so may fit in a reasonable internal memory even when $T$ is very large.  We show how the BWT can be constructed from a dictionary and parse alone.    Lastly, we give experiments demonstrating how the run length compressed FM-index can be constructed from the parse and dictionary.   The efficiency of this construction shows that this method opens up the possibility to cosntructing the BWT for datasets that are terabytes in size; such as GenomeTrakr \cite{genometrakr} and MetaSub \cite{metasub}.  

Finally, we note that when downloading large datasets, prefix-free parsing can avoid storing the whole uncompressed dataset in memory or on disk.  Suppose we run the parser on the dataset as it is downloaded, either as a stream or in chunks.  We have to keep the dictionary in memory for parsing but we can write the parse to disk as we go, and in any case we can use less total space than the dataset itself.  Ideally, the parsing could even be done server-side to reduce transmission time and/or bandwidth --- which we leave for future implementation and experimentation. 


\begin{backmatter}
\section*{Acknowledgements}
The authors thank Risto J\"arvinen for the insight they gained from his project on {\tt rsync} in the Data Compression course at Aalto University. 

\section*{Authors' contributions}
TG and GM conceptualized the idea and developed the algorithmic contributions of this work. AK and GM implemented the construction of the prefix-free parsing and conducted all experiments.  CB and TG assisted and oversaw the experiments and implementation.  TM and BL implemented and tested the construction of the run-length compressed FM-index.  All authors contributed to the writing of this manuscript.

\section*{Availability}
Prefix-free parsing and all accompanied documents are available at \url{https://gitlab.com/manzai/Big-BWT}.

\section*{Competing interests}
The authors declare that they have no competing interests.

\section*{Funding} CB and AK were supported by National Science Foundation (IIS 1618814). AK was also supported by a post-doctoral fellowship from the University of Florida Informatics Institute. TG  was partially supported by FONDECYT (1171058).  GM was partially supported by PRIN grant (201534HNXC).

\end{backmatter}


\newcommand{\BMCxmlcomment}[1]{}

\BMCxmlcomment{

<refgrp>

<bibl id="B1">
  <title><p>A global reference for human genetic variation</p></title>
  <aug>
    <au><cnm>{The 1000 Genomes Project Consortium}</cnm></au>
  </aug>
  <source>Nature</source>
  <pubdate>2015</pubdate>
  <volume>526</volume>
  <fpage>68</fpage>
  <lpage>-74</lpage>
</bibl>

<bibl id="B2">
  <title><p>The 100,000 Genomes Project: bringing whole genome sequencing to
  the NHS</p></title>
  <aug>
    <au><snm>Turnbull</snm><fnm>C.</fnm></au>
    <au><cnm>others</cnm></au>
  </aug>
  <source>British Medical Journal</source>
  <pubdate>2018</pubdate>
  <volume>361</volume>
  <fpage>k1687</fpage>
</bibl>

<bibl id="B3">
  <title><p>Whole-genome sequencing is taking over foodborne disease
  surveillance</p></title>
  <aug>
    <au><snm>Carleton</snm><fnm>H.A.</fnm></au>
    <au><snm>Gerner Smidt</snm><fnm>P.</fnm></au>
  </aug>
  <source>Microbe</source>
  <pubdate>2016</pubdate>
  <volume>11</volume>
  <fpage>311</fpage>
  <lpage>-317</lpage>
</bibl>

<bibl id="B4">
  <title><p>The Public Health Impact of a Publically Available, Environmental
  Database of Microbial Genomes</p></title>
  <aug>
    <au><snm>Stevens</snm><fnm>E.L.</fnm></au>
    <au><snm>Timme</snm><fnm>R.</fnm></au>
    <au><snm>Brown</snm><fnm>E.W.</fnm></au>
    <au><snm>Allard</snm><fnm>M.W.</fnm></au>
    <au><snm>Strain</snm><fnm>E.</fnm></au>
    <au><snm>Bunning</snm><fnm>K.</fnm></au>
    <au><snm>Musser</snm><fnm>S.</fnm></au>
  </aug>
  <source>Frontiers in Microbiology</source>
  <pubdate>2017</pubdate>
  <volume>8</volume>
  <fpage>808</fpage>
</bibl>

<bibl id="B5">
  <title><p>A Block-sorting Lossless Compression Algorithm</p></title>
  <aug>
    <au><snm>Burrows</snm><fnm>M</fnm></au>
    <au><snm>Wheeler</snm><fnm>DJ</fnm></au>
  </aug>
  <pubdate>1994</pubdate>
</bibl>

<bibl id="B6">
  <title><p>{Burrows-Wheeler} Transform for Terabases</p></title>
  <aug>
    <au><snm>Sir{\'{e}}n</snm><fnm>J</fnm></au>
  </aug>
  <source>Proccedings of the 2016 Data Compression Conference (DCC)</source>
  <pubdate>2016</pubdate>
  <fpage>211</fpage>
  <lpage>-220</lpage>
</bibl>

<bibl id="B7">
  <title><p>Lightweight Data Indexing and Compression in External
  Memory</p></title>
  <aug>
    <au><snm>Ferragina</snm><fnm>P</fnm></au>
    <au><snm>Gagie</snm><fnm>T</fnm></au>
    <au><snm>Manzini</snm><fnm>G</fnm></au>
  </aug>
  <source>Algorithmica</source>
  <pubdate>2012</pubdate>
  <volume>63</volume>
  <issue>3</issue>
  <fpage>707</fpage>
  <lpage>-730</lpage>
</bibl>

<bibl id="B8">
  <title><p>From {LZ77} to the Run-Length Encoded Burrows-Wheeler Transform,
  and Back</p></title>
  <aug>
    <au><snm>Policriti</snm><fnm>A</fnm></au>
    <au><snm>Prezza</snm><fnm>N</fnm></au>
  </aug>
  <source>Proceedings of the 28th Symposium on Combinatorial Pattern Matching
  (CPM)</source>
  <pubdate>2017</pubdate>
  <fpage>17:1</fpage>
  <lpage>-17:10</lpage>
</bibl>

<bibl id="B9">
  <source>https://rsync.samba.org</source>
</bibl>

<bibl id="B10">
  <title><p>Practical linear-time \emph{O}(1)-workspace suffix sorting for
  constant alphabets</p></title>
  <aug>
    <au><snm>Nong</snm><fnm>G</fnm></au>
  </aug>
  <source>{ACM} Trans. Inf. Syst.</source>
  <pubdate>2013</pubdate>
  <volume>31</volume>
  <issue>3</issue>
  <fpage>15</fpage>
</bibl>

<bibl id="B11">
  <title><p>Indexing compressed text</p></title>
  <aug>
    <au><snm>Ferragina</snm><fnm>P</fnm></au>
    <au><snm>Manzini</snm><fnm>G</fnm></au>
  </aug>
  <source>Journal of the ACM (JACM)</source>
  <publisher>ACM</publisher>
  <pubdate>2005</pubdate>
  <volume>52</volume>
  <issue>4</issue>
  <fpage>552</fpage>
  <lpage>-581</lpage>
</bibl>

<bibl id="B12">
  <title><p>Inducing enhanced suffix arrays for string collections</p></title>
  <aug>
    <au><snm>Louza</snm><fnm>FA</fnm></au>
    <au><snm>Gog</snm><fnm>S</fnm></au>
    <au><snm>Telles</snm><fnm>GP</fnm></au>
  </aug>
  <source>Theor. Comput. Sci.</source>
  <pubdate>2017</pubdate>
  <volume>678</volume>
  <fpage>22</fpage>
  <lpage>-39</lpage>
</bibl>

<bibl id="B13">
  <source>http://pizzachili.dcc.uchile.cl/repcorpus.html</source>
</bibl>

<bibl id="B14">
  <title><p>The Public Health Impact of a Publically Available, Environmental
  Database of Microbial Genomes</p></title>
  <aug>
    <au><snm>Stevens</snm><fnm>EL</fnm></au>
    <au><snm>Timme</snm><fnm>R</fnm></au>
    <au><snm>Brown</snm><fnm>EW</fnm></au>
    <au><snm>Allard</snm><fnm>MW</fnm></au>
    <au><snm>Strain</snm><fnm>E</fnm></au>
    <au><snm>Bunning</snm><fnm>K</fnm></au>
    <au><snm>Musser</snm><fnm>S</fnm></au>
  </aug>
  <source>Frontiers in Microbiology</source>
  <publisher>Frontiers</publisher>
  <pubdate>2017</pubdate>
  <volume>8</volume>
  <fpage>808</fpage>
</bibl>

<bibl id="B15">
  <title><p>Ultrafast and memory-efficient alignment of short {DNA} sequences
  to the human genome</p></title>
  <aug>
    <au><snm>Langmead</snm><fnm>B</fnm></au>
    <au><snm>Trapnell</snm><fnm>C</fnm></au>
    <au><snm>Pop</snm><fnm>M</fnm></au>
    <au><snm>Salzberg</snm><fnm>SL</fnm></au>
  </aug>
  <source>Genome biology</source>
  <publisher>BioMed Central</publisher>
  <pubdate>2009</pubdate>
  <volume>10</volume>
  <issue>3</issue>
  <fpage>R25</fpage>
</bibl>

<bibl id="B16">
  <title><p>Fast gapped-read alignment with {Bowtie 2}</p></title>
  <aug>
    <au><snm>Langmead</snm><fnm>B</fnm></au>
    <au><snm>Salzberg</snm><fnm>SL</fnm></au>
  </aug>
  <source>Nature Methods</source>
  <pubdate>2012</pubdate>
  <volume>9</volume>
  <issue>4</issue>
  <fpage>357</fpage>
  <lpage>-360</lpage>
  <url>http://www.nature.com/doifinder/10.1038/nmeth.1923
  http://www.ncbi.nlm.nih.gov/pubmed/22388286</url>
</bibl>

<bibl id="B17">
  <title><p>Fast and accurate long-read alignment with Burrows--Wheeler
  transform</p></title>
  <aug>
    <au><snm>Li</snm><fnm>H</fnm></au>
    <au><snm>Durbin</snm><fnm>R</fnm></au>
  </aug>
  <source>Bioinformatics</source>
  <publisher>Oxford University Press</publisher>
  <pubdate>2010</pubdate>
  <volume>26</volume>
  <issue>5</issue>
  <fpage>589</fpage>
  <lpage>-595</lpage>
</bibl>

<bibl id="B18">
  <title><p>SOAP2: an improved ultrafast tool for short read
  alignment</p></title>
  <aug>
    <au><snm>Li</snm><fnm>R</fnm></au>
    <au><snm>Yu</snm><fnm>C</fnm></au>
    <au><snm>Li</snm><fnm>Y</fnm></au>
    <au><snm>Lam</snm><fnm>TW</fnm></au>
    <au><snm>Yiu</snm><fnm>SM</fnm></au>
    <au><snm>Kristiansen</snm><fnm>K</fnm></au>
    <au><snm>Wang</snm><fnm>J</fnm></au>
  </aug>
  <source>Bioinformatics</source>
  <publisher>Oxford University Press</publisher>
  <pubdate>2009</pubdate>
  <volume>25</volume>
  <issue>15</issue>
  <fpage>1966</fpage>
  <lpage>-1967</lpage>
</bibl>

<bibl id="B19">
  <title><p>Optimal-Time Text Indexing in BWT-runs Bounded Space</p></title>
  <aug>
    <au><snm>Gagie</snm><fnm>T</fnm></au>
    <au><snm>Navarro</snm><fnm>G</fnm></au>
    <au><snm>Prezza</snm><fnm>N</fnm></au>
  </aug>
  <source>Proceedings of the 29th Symposium on Discrete Algorithms
  (SODA)</source>
  <pubdate>2018</pubdate>
  <fpage>1459</fpage>
  <lpage>-1477</lpage>
</bibl>

<bibl id="B20">
  <title><p>From Theory to Practice: Plug and Play with Succinct Data
  Structures</p></title>
  <aug>
    <au><snm>Gog</snm><fnm>S</fnm></au>
    <au><snm>Beller</snm><fnm>T</fnm></au>
    <au><snm>Moffat</snm><fnm>A</fnm></au>
    <au><snm>Petri</snm><fnm>M</fnm></au>
  </aug>
  <source>13th International Symposium on Experimental Algorithms, (SEA
  2014)</source>
  <pubdate>2014</pubdate>
  <fpage>326</fpage>
  <lpage>337</lpage>
</bibl>

<bibl id="B21">
  <title><p>A global reference for human genetic variation</p></title>
  <aug>
    <au><snm>Consortium</snm><fnm>TGP</fnm></au>
  </aug>
  <source>Nature</source>
  <pubdate>2015</pubdate>
  <volume>526</volume>
  <issue>7571</issue>
  <fpage>68</fpage>
  <lpage>-74</lpage>
  <url>https://www.nature.com/articles/nature15393</url>
</bibl>

<bibl id="B22">
  <title><p>{{B}{C}{F}tools/{R}o{H}: a hidden {M}arkov model approach for
  detecting autozygosity from next-generation sequencing data}</p></title>
  <aug>
    <au><snm>Narasimhan</snm><fnm>V.</fnm></au>
    <au><snm>Danecek</snm><fnm>P.</fnm></au>
    <au><snm>Scally</snm><fnm>A.</fnm></au>
    <au><snm>Xue</snm><fnm>Y.</fnm></au>
    <au><snm>Tyler Smith</snm><fnm>C.</fnm></au>
    <au><snm>Durbin</snm><fnm>R.</fnm></au>
  </aug>
  <source>Bioinformatics</source>
  <pubdate>2016</pubdate>
  <volume>32</volume>
  <issue>11</issue>
  <fpage>1749</fpage>
  <lpage>-1751</lpage>
</bibl>

<bibl id="B23">
  <title><p>{The Metagenomics and Metadesign of the Subways and Urban Biomes
  (MetaSUB) International Consortium inaugural meeting report}</p></title>
  <aug>
    <au><cnm>{MetaSUB International Consortium}</cnm></au>
  </aug>
  <source>Microbiome</source>
  <pubdate>2016</pubdate>
  <volume>4</volume>
  <issue>1</issue>
  <fpage>24</fpage>
</bibl>

</refgrp>
} 

\end{document}